\documentclass[sigconf]{acmart}
\usepackage[utf8]{inputenc}
\usepackage{graphicx}
\usepackage{amsthm, amsmath, amssymb, amsfonts}
\usepackage{subfloat}
\usepackage{subfig}
\usepackage{caption}
\usepackage{booktabs} % For formal tables

\usepackage{multirow}
\setcopyright{none}
\renewcommand\footnotetextcopyrightpermission[1]{} % removes footnote with conference information in first column
\newtheorem{definition}{Definition}
\newtheorem{problem}{Problem}
\newtheorem*{approach}{Approach}
\newtheorem*{remark}{Remark}

\newcommand\blfootnote[1]{%
  \begingroup
  \renewcommand\thefootnote{}\footnote{#1}%
  \addtocounter{footnote}{-1}%
  \endgroup
}

\begin{document}

\title{Verisig: verifying safety properties of hybrid systems with neural network controllers}
%\subtitle{A Hybrid-System Approach}
\author{Radoslav Ivanov}
\affiliation{%
  \institution{Department of Computer and Information Science\\ University of Pennsylvania}
  \streetaddress{3330 Walnut Street}
  \city{Philadelphia}
  \state{Pennsylvania}
}
\email{rivanov@seas.upenn.edu}

\author{James Weimer}
\affiliation{%
  \institution{Department of Computer and Information Science\\ University of Pennsylvania}
  \streetaddress{3330 Walnut Street}
  \city{Philadelphia}
  \state{Pennsylvania}
}
\email{weimerj@seas.upenn.edu}

\author{Rajeev Alur}
\affiliation{%
  \institution{Department of Computer and Information Science\\ University of Pennsylvania}
  \streetaddress{3330 Walnut Street}
  \city{Philadelphia}
  \state{Pennsylvania}
}
\email{alur@cis.upenn.edu}

\author{George J. Pappas}
\affiliation{%
  \institution{Department of Electrical and Systems Engineering\\ University of Pennsylvania}
  \streetaddress{3330 Walnut Street}
  \city{Philadelphia}
  \state{Pennsylvania}
}
\email{pappasg@seas.upenn.edu}

\author{Insup Lee}
\affiliation{%
  \institution{Department of Computer and Information Science\\ University of Pennsylvania}
  \streetaddress{3330 Walnut Street}
  \city{Philadelphia}
  \state{Pennsylvania}
}
\email{lee@cis.upenn.edu}

\begin{abstract}
This paper presents Verisig, a hybrid system approach to verifying
safety properties of closed-loop systems using neural networks as
controllers. Although techniques exist for verifying input/output
properties of the neural network itself, these methods cannot be used
to verify properties of the closed-loop system (since they work with
piecewise-linear constraints that do not capture non-linear plant
dynamics). To overcome this challenge, we focus on sigmoid-based
networks and exploit the fact that the sigmoid is the solution to a
quadratic differential equation, which allows us to transform the
neural network into an equivalent hybrid system. By composing the
network's hybrid system with the plant's, we transform the problem
into a hybrid system verification problem which can be solved using
state-of-the-art reachability tools. We show that reachability is
decidable for networks with one hidden layer and decidable for general
networks if Schanuel's conjecture is true. We evaluate the
applicability and scalability of Verisig in two case studies, one from
reinforcement learning and one in which the neural network is used to
approximate a model predictive controller.\blfootnote{This material is
  based upon work supported by the Air Force Research Laboratory
  (AFRL) and the Defense Advanced Research Projects Agency (DARPA)
  under Contract No. FA8750-18-C-0090. Any opinions, findings and
  conclusions or recommendations expressed in this material are those
  of the author(s) and do not necessarily reflect the views of the
  AFRL, DARPA, the Department of Defense, or the United States
  Government. This work was supported in part by NSF grant
  CNS-1837244. This research was supported in part by ONR
  N000141712012.}
\end{abstract}

\maketitle

%!TEX root = main.tex

\section{Introduction}
\label{sec:intro}
In recent years, deep neural networks (DNNs) have been successfully
applied to multiple challenging tasks such as image
processing~\cite{taigman14}, reinforcement learning~\cite{mnih15},
learning model predictive controllers (MPCs)~\cite{royo18}, natural
language translation~\cite{sutskever14}, and games such as
Go~\cite{silver16}. These promising results have inspired system
developers to use DNNs in
%increasingly complex and \JW{removed -- increasingly complex isn't
%really the motivation here.}
safety-critical Cyber-Physical Systems (CPS) such as autonomous
vehicles~\cite{uber_report} and air traffic collision avoidance
systems~\cite{julian16}. At the same time, several recent incidents
(e.g., Tesla~\cite{tesla_report} and Uber~\cite{uber_report}
autonomous driving crashes) have underscored the need to better
understand DNNs and verify safety properties about CPS using such
networks.

The traditional way of assessing a learning algorithm's performance is
through bounding the expected generalization error (EGE) of a trained
classifier, i.e., the expected difference between the classifier's
error on training versus test examples~\cite{mohri12}. The EGE can be
usually bounded (e.g., in a probably approximately correct
sense~\cite{kearns94}) by assuming that a large enough training set
satisfying
%a certain \JW{I changed this sentence a bit ... changed a "certain"
%to "some", changed requirements to assumptions, and assumptions
%plural}
some statistical assumptions (e.g., independent and identically
distributed examples) is available. However, it is difficult to obtain
tight EGE bounds for DNNs due to the high-dimensional input and
parameter settings DNNs are used in (e.g., thousands of inputs, such
as pixels in an image, and millions of
parameters)~\cite{zhang16}. Thus, it remains a challenge to bound the
classification error of DNNs used in real-world applications; in fact,
several robustness issues with DNNs have been discovered (e.g.,
adversarial examples~\cite{szegedy13}).

As an alternative way of assuring the safety of systems using DNNs,
researchers have focused on analyzing the \emph{trained DNNs} used in
specific
systems~\cite{dreossi17,dutta18,ehlers17,katz17,xiang17,xiang18}. While
analytic proofs of input/output properties are hard to obtain due to
the complexity of DNNs (namely, they are universal function
approximators~\cite{hornik89}), prior work has shown it is possible to
formally verify properties about DNNs by adapting existing
satisfiability modulo theory (SMT) solvers~\cite{ehlers17,katz17} and
mixed-integer linear program (MILP) optimizers~\cite{dutta18}. In
particular, these techniques can verify linear properties about the
DNN's output given linear constraints on the inputs. These approaches
exploit the piecewise-linear nature of the rectified linear units
(ReLUs) used in many DNNs and scale well by encoding the DNN as an
input to efficient SMT/MILP solvers. As a result, existing tools can
be used on reasonably sized DNNs, i.e., DNNs with several layers and a
few hundred neurons per layer.
%% %% \JW{I would avoid saying exactly how deep / many neurons
%% %% they can handle as part of the introduction -- unless you are
%% %% absolutely sure you have the numbers right (see last year's
%% %% CAV).} \RI{This is the right order of magnitude. If I don't give
%% %% order of magnitude, people will be asking what is ``reasonably
%% %% sized''.}

Although the SMT- and MILP-based approaches work well for the
verification of properties of the DNN itself, these techniques cannot
be straightforwardly extended to closed-loop systems using DNNs as
controllers. Specifically, the non-linear dynamics of a typical CPS
plant cannot be captured by these frameworks except for special cases
such as discrete-time linear systems. While it is in theory possible
to also approximate the plant dynamics with a ReLU-based DNN and
verify properties about it, it is not clear how to relate properties
of the approximating system to properties of the actual plant. As a
result, it is challenging to use existing techniques to reason about
the safety of the overall system.
%% \JW{So this segue paragraph has 1 issue.  Why can't I just
%% approximate the non-linear plant dynamics as a DNN with ReLUs?
%% This needs to be addressed because people will wonder why that
%% approach wasn't pursued.}

To overcome this limitation, we investigate an alternative approach,
named Verisig, that allows us to verify properties of the closed-loop
system. In particular, we consider CPS using sigmoid-based DNNs
instead of ReLU-based ones and use the fact that the sigmoid is the
solution to a quadratic differential equation.
%% \IL{Emphasize/state that this is the first paper to work on
%%   sigmoid-based DNN verification (using hybrid systems)?}
This allows us to transform the DNN into an equivalent hybrid system
such that a DNN with $L$ layers and $N$ neurons per layer can be
represented as a hybrid system with $L+1$ modes and $2N$ states. In
turn, we compose the DNN's hybrid system with the plant's and verify
properties of the composed system's reachable space by using existing
reachability tools such as dReach~\cite{kong15} and
Flow*~\cite{chen13}.
%% The mildly non-linear form of the differential
%% equations, coupled with the bounded number of mode transitions,
%% suggests that it is possible to verify properties of sigmoid-based
%% DNNs of similar sizes as the ReLU-based ones mentioned above.
We emphasize that this paper is not only the first to verify
properties about closed-loop systems with DNN controllers, but also
\emph{the first to consider sigmoid-based DNN verification using
  hybrid systems}.
%% (e.g., Flow* scales to systems with a few hundred states as
%% shown in Section~\ref{sec:comparison}). \JW{Again, I think we should
%%   avoid actual numbers in the intro -- can we just say "similar size"
%%   and leave it at that.}

To analyze the feasibility of the proposed approach, we show that DNN
reachability (i.e., checking whether the DNN's outputs lie in some set
given constraints on the inputs) can be transformed into a
real-arithmetic property with transcendental functions, which is
decidable if Schanuel's conjecture is true~\cite{wilkie97}. We also
prove that reachability is decidable for DNNs with one hidden layer,
given interval constraints on the inputs. Finally, by casting the
problem in the dReach framework, we also show that reachability is
$\delta$-decidable for general DNNs~\cite{gao14}.
%% \JW{This is a fine way of introducing the second contribution
%%   ... but the second half of it is lost in the the list of
%%   contributions in the final paragraph below.  It IS a contribution
%%   that you're able to formulate the problem in a way that fits the
%%   dReach input format.} \RI{I rewrote the paragraph below to
%%   reflect this contribution.}

To evaluate the applicability of Verisig, we consider two case
studies, one from reinforcement learning (RL) and one where a DNN is
used to approximate an MPC with safety guarantees. DNNs are
increasingly being used in both of these domains, so it is essential
to be able to verify properties of interest about the closed-loop
system. We trained a DNN for a benchmark RL task, Mountain Car, and
verified that the DNN will achieve its control task (i.e., drive an
underpowered car up a hill) within the problem constraints. In the MPC
approximation setting, we used an existing technique to approximate an
MPC with a DNN~\cite{royo18} and verified that a DNN-controlled
quadrotor will reach its goal without colliding into obstacles.

Finally, we evaluate the scalability of Verisig as used with Flow* by
training DNNs of increasing size on the Mountain Car problem. For each
DNN, we record the time it takes to compute the output reachable
set. For comparison purposes, we implemented a piecewise-linear
approach to approximate each sigmoid as suggested in prior
work~\cite{dutta18}; in this setting, the problem is cast as an MILP
program that can be solved by an MILP optimizer such as
Gurobi~\cite{gurobi}. We observe that, at similar levels of
approximation, the MILP-based approach is faster than Verisig+Flow*
for small DNNs and DNNs with few layers. However, the MILP-based
approach's runtimes increase exponentially for deeper networks whereas
Verisig+Flow* scales linearly with the number of layers since the same
computation is run in each mode (each layer). This is another positive
feature of our technique since deeper networks are known to learn more
efficiently than shallow ones~\cite{rolnick17,telgarsky16}.

In summary, this paper has three contributions: 1) we develop an
approach to transform a DNN into a hybrid system, which allows us to
cast the closed-loop system verification problem into a hybrid system
verification problem; 2) we show that the DNN reachability problem is
decidable for DNNs with one hidden layer and decidable for general
DNNs if Schanuel's conjecture holds;
%% and a proof that the $\delta$-reachability problem is decidable for DNNs with an arbitrary number of layers;\RI{I removed this claim after I couldn't figure out how to formalize the delta-reachability claim.}
3) we evaluate both the applicability and scalability of Verisig using
two case studies.

The rest of this paper is organized as
follows. Section~\ref{sec:problem} states the problem addressed in
this work. Section~\ref{sec:decidability} analyzes the decidability of
the verification problem, and Section~\ref{sec:approach} describes
Verisig. Sections~\ref{sec:results} and~\ref{sec:comparison} present
the case study evaluations in terms of applicability and
scalability. Section~\ref{sec:discussion} provides concluding remarks.
%\IL{Add: The rest of the paper is organized as follows....}
%% \JW{The contributions are still weird w.r.t. the previous
%% paragraph.  There you say that deeper networks are better and that
%% justifies using Verisig -- here you say we prove something about a
%% very shallow network.  I think you need generalize the second
%% contribution or downplay the deep network stuff as a motivation for
%% our appraoch.} \RI{I think this is better. I used ``show'' instead
%% of prove, which works for both I think. Either way, I don't think
%% that this is a major issue -- we haven't proved it's decidable for
%% many hidden layers but it's probably still true.}

%% Although ReLU-based networks are more popular due to
%% faster training time and easier parameter
%% initialization~\cite{krizhevsky12}, sigmoid-based DNNs are equally
%% powerful if properly trained~\cite{ioffe15} and potentially more
%% robust to input perturbations due to their smoothness (it is believed
%% that ReLU-based DNNs suffer from adversarial examples due to
%% ``pockets'' in the network's decision boundary caused by the
%% non-differentiable nature of ReLUs~\cite{szegedy13}). Thus,
%% sigmoid-based DNNs can be more useful in safety-critical applications
%% where robust performance is essential.

%!TEX root = main.tex

\section{Problem Formulation}
\label{sec:problem}
\begin{figure}[!t]
  \centering 
  \includegraphics[width=0.97\linewidth]{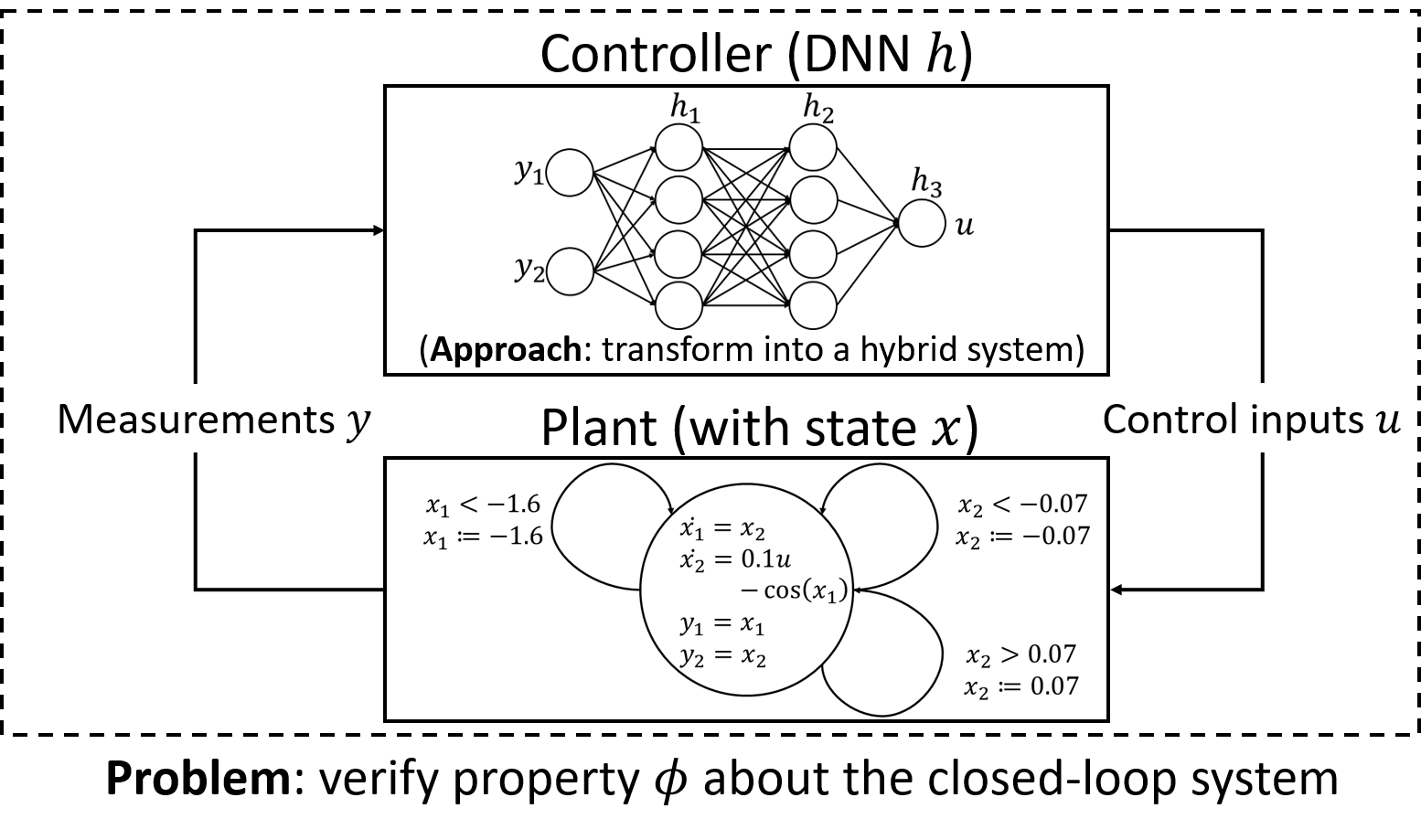}
  \caption{Illustration of the closed-loop system considered in this
    paper. The plant model is given as a standard hybrid system,
    whereas the controller is a DNN. The problem is to verify a
    property of the closed-loop system.}
  \label{fig:cl_system}
\end{figure}

This section formulates the problem considered in this paper. We
consider a closed-loop system, as shown in Figure~\ref{fig:cl_system},
with states $x$, measurements $y$, and a controller $h$. The states
and measurements are formalized in the next subsection, followed by
the (DNN) controller description and the problem statement itself,
i.e., the verification of a property $\phi$ about the closed-loop
system.

\subsection{Plant Model}
\label{sec:plant_model}
We assume the plant dynamics are given as a hybrid system.  A hybrid
system's state space consists of a finite set of discrete modes and a
finite number of continuous variables~\cite{lafferriere99}. Within
each mode, continuous variables evolve according to known differential
equations; we focus specifically on differential equations with
respect to time. Furthermore, each mode contains a set of invariants
that hold true while the system is in that mode. Transitions between
modes are controlled by guards, which represent conditions on the
continuous variables. Finally, continuous variables can be reset
during each mode transition. The formal definition is provided below.

\begin{definition}[Hybrid System]
A hybrid system with inputs $u$ and outputs $y$ is a tuple $H = (X,
X_0, F, E, I, G, R, g)$ where
\begin{itemize}

\item $X = X_D \times X_C$ is the state space with $X_D = \{q_1,
  \dots, q_m\}$ and $X_C$ a manifold;

\item $X_0 \subseteq X$ is the set of initial states;

\item $F: X \longrightarrow TX_C$ assigns to each discrete mode $q \in
  X_D$ a vector field $f_q$, i.e., $\dot{x}_c = f_q(x_c, u)$
  in mode $q$;

\item $E \subseteq X_D \times X_D$ is the set of mode transitions;

\item $I: X_D \longrightarrow 2^{X_C}$ assigns to $q \in X_D$ an
  invariant of the form $I(q) \subseteq X_C$;
  
\item $G: E \longrightarrow 2^{X_C}$ assigns to each edge $e = (q_1,
  q_2)$ a guard $U \subseteq I(q_1)$;
  
\item $R: E \longrightarrow 2^{X_C}$ assigns to each edge $e =
  (q_1, q_2)$ a reset $V \subseteq I(q_2)$;

\item $g: X \longrightarrow \mathbb{R}^p$ is the observation model
  such that $y = g(x)$.
  
\end{itemize}
\end{definition}

\subsection{DNN Controller Model}
\label{sec:dnn_model}
As mentioned in Section~\ref{sec:intro}, the controller is implemented
by a DNN. To simplify the presentation, we assume the DNN is a
feedforward neural network. However, the proposed technique applies to
all common classes such as convolutional, residual or recurrent DNNs.

A DNN controller maps measurements $y$ to control inputs $u$ and can
be defined as a function $h$ as follows: $h: \mathbb{R}^p \rightarrow
\mathbb{R}^q$. As illustrated in Figure~\ref{fig:cl_system}, a typical
DNN has a layered architecture and can be represented as a composition
of its $L$ layers: $$h(y) = h_L \circ h_{L-1} \circ \dots \circ h_1
(y),$$ where each hidden layer $h_i$, $i \in \{1, \dots, L-1\}$, has
an element-wise (with each element called a neuron) non-linear
activation function: $$h_i(y) = a(W_iy + b_i).$$ Each $h_i$ is
parameterized by a weight matrix $W_i$ and an offset vector $b_i$. The
most common types of activation functions are
\begin{itemize}

\item ReLU: $a(y) := ReLU(y) = \max\{0,y\}$,

\item sigmoid: $a(y) := \sigma(y) = \frac{1}{1 + e^{-y}}$,

\item hyperbolic tangent: $a(y) := \text{tanh}(y) = \frac{e^y -
  e^{-y}}{e^y + e^{-y}}$.
  
\end{itemize}
%
%As argued in the introduction, we consider sigmoid and tanh activation
%functions (which also fall in the broad class of sigmoidal functions)
%in this paper, different from most existing works that assume ReLU
%activation functions.
%
As argued in the introduction, and different from most existing works
that assume ReLU activation functions, this work considers sigmoid and
tanh activation functions (which also fall in the broad class of
sigmoidal functions). Finally, the last layer $h_L$ is
linear:\footnote{The last layer is by convention a linear layer,
  although it could also have a non-linear activation, as shown in the
  Mountain Car case study.} $$h_L(y) = W_Ly + b_L,$$ which is
parameterized by a matrix $W_L$ and a vector $b_L$.

During training, the parameters $(W_1, b_1, \dots, W_{L}, b_{L})$ are
learned through an optimization algorithm (e.g., stochastic gradient
descent~\cite{goodfellow16}) used on a training set. In this paper, we
assume the DNN \emph{is already trained}, i.e., all parameters are
known and fixed.

\subsection{Problem Statement}

Given the plant model and the DNN controller model described in this
section, we identify two verification problems. The first one is the
reachability problem for the DNN itself.

\begin{problem}
\label{prob:dnn_prob}
Let $h$ be a DNN as described in Section~\ref{sec:dnn_model}. The DNN
verification problem, expressed as property $\phi_{\text{dnn}}$, is to
verify a property $\psi_{\text{dnn}}$ on the DNN's outputs $u$ given
constraints $\xi_{\text{dnn}}$ on the inputs $y$:
\begin{align}
\phi_{\text{dnn}}(y, u) \equiv (\xi_{\text{dnn}}(y) \wedge h(y) = u) \Rightarrow \psi_{\text{dnn}}(u).
\end{align}
\end{problem}

Problem~\ref{prob:cl_prob} is to verify a property of the closed-loop
system.

\begin{problem}
\label{prob:cl_prob}
Let $S = h\ ||\ H_P$ be the composition of a DNN controller $h$
(Section~\ref{sec:dnn_model}) and a plant $P$, modeled with a hybrid
system $H_P$ (Section~\ref{sec:plant_model}). Given a property $\xi$
on the initial states $X_0$ of $P$, the problem, expressed as property
$\phi$, is to verify a property $\psi$ of the reachable states of $P$:
\begin{equation}
\label{eq:cl_prop}
\phi(X_0, x(t)) \equiv \xi(X_0) \Rightarrow \psi(x(t)),\ \forall t \ge 0.
\end{equation}
\end{problem}
Our approach to Problem~\ref{prob:dnn_prob}, namely transforming the
DNN into an equivalent hybrid system, also presents a solution to
Problem~\ref{prob:cl_prob} since we can compose the DNN's hybrid
system with the plant's and can use existing hybrid system
verification tools.

\begin{approach}
We approach Problem~\ref{prob:dnn_prob} by transforming $h$ into an
equivalent hybrid system $H_h$ such that if $x_0$ is an initial
condition of $H_h$, then the only reachable state is
$h(x_0)$. Problem~\ref{prob:cl_prob} is addressed by verifying
properties about the composed hybrid system $H_h\ ||\ H_P$.
\end{approach}

%!TEX root = main.tex

\section{On the Decidability of Sigmoid-Based DNN Reachability}
\label{sec:decidability}
Before describing our approach to the problems stated in
Section~\ref{sec:problem}, a natural question to ask is whether these
problems are decidable. The answer is not obvious due to the
non-linear nature of the sigmoid. This section shows that if the DNN's
inputs and outputs are given as a real-arithmetic property, then
reachability can be stated as a real-arithmetic property with
transcendental functions, which is decidable if Schanuel's conjecture
is true~\cite{wilkie97}. Furthermore, we prove decidability for the
case of DNNs with a single hidden layer, under mild assumptions on the
DNN parameters. Finally, we argue that by casting the DNN verification
problem into a hybrid system verification problem, we obtain a
$\delta$-decidable problem~\cite{gao14}.\footnote{Note that the
  results presented in this section hold for DNNs with sigmoid
  activation functions, but similar results can be shown for tanh.}

\subsection{DNNs with multiple hidden layers}

As formalized in Section~\ref{sec:problem}, the reachability property
of a DNN $h$ with inputs $y$ and outputs $u$ has the general form:
\begin{align}
\label{eq:reach_prop}
\phi(y, u) \equiv (\xi(y) \wedge h(y) = u) \Rightarrow \psi(u),
\end{align}
where $\xi$ and $\psi$ are given properties on the real
numbers. Verifying properties on the real numbers is undecidable in
general. A notable exception is first-order logic formulas over
$(\mathbb{R}, <, +, -, \cdot, 0, 1)$, i.e., the language where $<$ is
the relation, +, -, and $\cdot$ are functions, and 0 and 1 are the
constants~\cite{tarski98}; we denote such formulas by
$\mathcal{R}$-formulas. Intuitively, $\mathcal{R}$-formulas are
first-order logic statements where the constraints are polynomial
functions of the variables with integer coefficients. Example
$\mathcal{R}$-formulas are $\forall x\ \forall y\ :\ xy > 0, \exists
x\ :\ x^2 - 2 = 0$, and $\exists w\ :\ xw^2 + yw + z = 0$.

Another relevant language is $(\mathbb{R}, <, +, -, \cdot, \text{exp},
0, 1)$, which also includes exponentiation; we denote these formulas
by $\mathcal{R}_{\text{exp}}$-formulas. Although it is an open
question whether verifying $\mathcal{R}_{\text{exp}}$-formulas is
decidable, it is known that decidability is connected to Schanuel's
conjecture~\cite{wilkie97}. Schanuel's conjecture concerns the
transcendence degree of certain field extensions of the rational
numbers and, if true, would imply that verifying
$\mathcal{R}_{\text{exp}}$-formulas is decidable~\cite{wilkie97}.

We focus on the case where $\xi$ and $\psi$ are
$\mathcal{R}$-formulas. The exponentiation in the sigmoid means that
$\phi$, however, is not a $\mathcal{R}$-formula. We show below that
$\phi$ is in fact an $\mathcal{R}_{\text{exp}}$-formula, which implies
that DNN reachability is decidable if Schanuel's conjecture is
true~\cite{wilkie97}.

\begin{proposition}
\label{prop:dnn_multi}
Let $h: \mathbb{R}^p \rightarrow \mathbb{R}^q$ be a sigmoid-based DNN
with $L - 1$ hidden layers (with $N$ neurons each) and rational
parameters. The property $\phi(y, u) \equiv (\xi(y) \wedge h(y) = u)
\Rightarrow \psi(u)$, where $\xi$ and $\psi$ are
$\mathcal{R}$-formulas, is an $\mathcal{R}_{\text{exp}}$-formula.
\end{proposition}

\begin{proof}
Since $\psi$ is an $\mathcal{R}$-formula, it suffices to show that
$\phi_0(y, u) \equiv \xi(y) \wedge h(y) = u$ can be expressed as an
$\mathcal{R}_{\text{exp}}$-formula. Note that
\begin{align*}
\phi_0(y, u) \equiv\ &\xi(y) \wedge h_1^1 = \frac{1}{1 + \text{exp}\{-(w_1^1)^{\top}y - b_1^1\}} \wedge \dots\\
&\wedge h_1^N = \frac{1}{1 + \text{exp}\{-(w^N_{1})^{\top}y - b^N_1\}} \wedge \dots\\
& \wedge h_{L-1}^{1} = \frac{1}{1 + \text{exp}\{-(w^1_{L-1})^{\top}h_{L-2} - b^1_{L-1}\}} \wedge \dots\\
&\wedge h_{L-1}^{N} = \frac{1}{1 + \text{exp}\{-(w^N_{L-1})^{\top}h_{L-2} - b^N_{L-1}\}}\\
&\wedge u = W_L [h_{L-1}^{1}, \dots, h_{L-1}^{N}]^{\top} + b_L,
\end{align*}
where $(w_i^j)^{\top}$ is row $j$ of $W_i$, and $h_l = [h_l^1, \dots,
  h_l^N]^{\top}, l \in \{1, \dots, L-1\}$. The last constraint, call
it $p(u)$, is already an $\mathcal{R}$-formula. Let $[W_i]_{jk} =
p_{jk}^i/q_{jk}^i$, with $p_{jk}^i$ and $q_{jk}^i > 0$ integers, and
let $d_0 = q_{11}^1q_{12}^1\cdots q_{Np}^{L-1}$. To remove fractions
from the exponents, we add extra variables $z_i$ and $v_i^j$ and
arrive at an equivalent property $\phi_{\mathbb{Z}}$, which is an
$\mathcal{R}_{\text{exp}}$-formula since all denominators are
$\mathcal{R}_{\text{exp}}$-formulas:
\begin{align*}
&\phi_{\mathbb{Z}}(y, u) \equiv\ \xi(y) \wedge z_0d_0 = y \wedge h_1^1 = \frac{1}{1 + \text{exp}\{-(r_1^1)^{\top}z_0 - v_1^1\}} \wedge \dots\\
&\wedge h_1^N = \frac{1}{1 + \text{exp}\{-(r^N_{1})^{\top}z_0 - v^N_1\}} \wedge v_1^1 = b_1^1 \wedge \dots \wedge v_1^N = b_1^N \wedge \dots\\
& \wedge z_{L-2}d_0 = h_{L-2} \wedge h_{L-1}^{1} = \frac{1}{1 + \text{exp}\{-(r^1_{L-1})^{\top}z_{L-2} - v^1_{L-1}\}} \wedge \dots\\
&\wedge h_{L-1}^{N} = \frac{1}{1 + \text{exp}\{-(r^N_{L-1})^{\top}z_{L-2} - v^N_{L-1}\}} \\
&\wedge v_{L-1}^1 = b_{L-1}^1 \wedge \dots \wedge v_{L-1}^N = b_{L-1}^N \wedge p(u),
\end{align*}
where $r_i^j = w_i^jd_0$ are vectors of integers; $v_i^j = b_i^j$ are
$\mathcal{R}$-formulas since $b_i^j$ are rational.
\end{proof}

\begin{corollary}[\cite{wilkie97}]
If Schanuel's conjecture holds, then verifying the property $\phi(y,
u) \equiv (\xi(y) \wedge h(y) = u) \Rightarrow \psi(u)$ is decidable
under the conditions stated in Proposition~\ref{prop:dnn_multi}.
\end{corollary}

\begin{remark}
Note that by transforming the DNN into an equivalent hybrid system (as
described in Section~\ref{sec:approach}), we show that DNN
reachability is $\delta$-decidable as well~\cite{gao14}. Intuitively,
$\delta$-decidability means that relaxing all constraints by a
rational $\delta$ results in a decidable problem; as shown in prior
work~\cite{gao14}, reachability is $\delta$-decidable for hybrid
systems with dynamics given by Type 2 computable functions. Since the
sigmoid is a Type 2 computable function, we have strong evidence to
believe that the proposed technique is a promising approach.
\end{remark}

\subsection{DNNs with a single hidden layer}

Regardless of whether Schanuel's conjecture holds, we can show that
DNN reachability is decidable for DNNs with a single hidden layer. In
particular, assuming interval bounds are given for each input, it is
possible to transform the reachability property into an
$\mathcal{R}$-formula, thus showing that verifying reachability is
decidable.

\begin{theorem}
\label{thm:decidable}
Let $h: \mathbb{R}^p \rightarrow \mathbb{R}^q$ be a sigmoid-based DNN
with one hidden layer (with $N$ neurons), i.e., $h(x) =
W_2(\sigma(W_1x + b_1)) + b_2$. Let $[W_1]_{ij} = p_{ij}/q_{ij}$ be
all rational and let $d_0 = q_{11}q_{12}\cdots q_{Np}$. Consider the
property
\begin{align*}
\phi(y, u) \equiv (\exists y \in I_y \wedge u = h(y)) \Rightarrow \psi(u),
\end{align*}
where $y = [y_1, \dots, y_p]^{\top} \in \mathbb{R}^p$, $u = [u_1,
  \dots, u_q]^{\top} \in \mathbb{R}^q$, $\psi$ is an
$\mathcal{R}$-formula, and $I_y = [\alpha_1, \beta_1] \times \dots
\times [\alpha_p, \beta_q] \subseteq \mathbb{R}^p$, i..e., the
Cartesian product of $p$ one-dimensional intervals. Then verifying
$\phi(y, u)$ is decidable if $e^{b_1^i}$, $e^{\alpha_j/d_0}$, and
$e^{\beta_j/d_0}$ are rational numbers for all $i \in \{1, \dots, N\}$
and $j \in \{1, \dots, p\}$ ($b_1^i$ is element $i$ of vector $b_1$).
\end{theorem}

\begin{proof}
The proof technique borrows ideas from~\cite{lafferriere99}. It
suffices to show that $\phi(y, u)$ is an $\mathcal{R}$-formula. Since
$\psi(u)$ is an $\mathcal{R}$-formula, we focus on the remaining part
of $\phi(y, u)$, call it $\phi_0(y, u)$. Then
\begin{align*}
\phi_0&(y, u) \equiv \exists y \in I_y \wedge h^1_1 = \frac{1}{1 +
  \text{exp}\{-(w_1^1)^{\top}y - b_1^1\}} \wedge \dots\\
&\wedge h^N_1 = \frac{1}{1 + \text{exp}\{-(w^N_{1})^{\top}y - b^N_1\}} \wedge u = W_2 [h^1_1, \dots, h^N_1]^{\top} + b_2,
\end{align*}
where $(w_1^i)^{\top}$ is row $i$ of $W_1$. Note that the last
constraint in $\phi_0(y, u)$, call it $p(u)$, is an
$\mathcal{R}$-formula. To remove fractions from the exponentials, we
change the limits of $y$. Consider the property
\begin{align*}
\phi_{\mathbb{Z}}(y,u) &\equiv \exists y \in I_y^{\mathbb{Z}} \wedge h^1_1 = \frac{1}{1 + \text{exp}\{-(r^1_1)^{\top}y - b_1^1\}} \wedge \dots\\
&\wedge h^N_1 = \frac{1}{1 + \text{exp}\{-(r_1^N)^{\top}y - b_1^N\}} \wedge p(u),
\end{align*}
where $I_y^{\mathbb{Z}} = [\alpha_1/d_0, \beta_1/d_0] \times \dots
\times [\alpha_p/d_0, \beta_p/d_0]$ and each $r_1^i = d_0w_1^i$ is a
vector of integers. Note that $\phi_0(y, u) \equiv
\phi_{\mathbb{Z}}(y, u)$, since a change of variables $z = y/d_0$
implies that $z \in I_y^{\mathbb{Z}}$ iff $y \in I_y$. To remove
exponentials from the constraints, we use their monotonicity property
and transform $\phi_{\mathbb{Z}}(x,y)$ into an equivalent property
$\phi_e(x,y)$:
\begin{align*}
\phi_e(y, u) &\equiv \exists y \in I_y^e \wedge h_1^1 = \frac{1}{1 +
  y_1^{r^1_{11}}\cdots y_p^{r^1_{1p}}\text{exp}\{- b_1^1\}} \wedge \dots\\
&\wedge h^N_1 = \frac{1}{1 + y_1^{r^N_{11}}\cdots y_p^{r^N_{1p}}\text{exp}\{- b^N_1\}} \wedge p(u),
\end{align*}
where $I_y^{e} = [e^{-\beta_1/d_0}, e^{-\alpha_1/d_0}] \times \dots
\times [e^{-\beta_p/d_0}, e^{-\alpha_p/d_0}]$, and $r_{1j}^i$ is
element $j$ of $r_1^i$. To see that $\phi_e(y, u) \equiv
\phi_{\mathbb{Z}}(y, u)$, take any $y \in I_y^{\mathbb{Z}}$ and note
that $\text{exp}\{-r^i_{1j}y_j\} = z_j^{r^i_{1j}}$, with $z_j =
e^{-y_j}$; thus, $z \in I_x^e$.

The final step transforms the property $\phi_e(y, u)$ into an
equivalent property $\nu(y, u)$ to eliminate negative integers
$r^i_{1j}$ in the exponents:
\begin{align*}
  \nu(y, u) &\equiv \exists y \in I_y^e\ \exists z \in I_y^{e-}\ y_1z_1 = 1 \wedge \dots \wedge y_p z_p = 1\\
  &\wedge h_1^1
= \frac{1}{1 + \displaystyle\prod_{j \in \mathcal{I}_{1}^+}y_j^{r^1_{1j}}\prod_{j \in \mathcal{I}_{1}^-}z_j^{-r^1_{1j}}\text{exp}\{-
  b_1^1\}} \wedge \dots\\ &\wedge h^N_1 = \frac{1}{1 +
  \displaystyle\prod_{j \in \mathcal{I}_{N}^+}y_j^{r^N_{1j}}\prod_{j \in \mathcal{I}_{N}^-}z_j^{-r^N_{1j}}\text{exp}\{- b_1^N\}} \wedge
p(u),
\end{align*}
where $I_y^{e-} = [e^{\alpha_1/d_0}, e^{\beta_1/d_0}] \times \dots
\times [e^{\alpha_p/d_0}, e^{\beta_p/d_0}]$, $\mathcal{I}_i^+ = \{k
\mid r^i_{1k} \ge 0\}$, and $\mathcal{I}_i^- = \{k \mid r^i_{1k} <
0\}$. Note that $\phi_e(y, u) \equiv \nu(y, u)$ since for $r^i_{1j} <
0$, the constraint $z_jy_j = 1$ implies $y_j^{r^i_{1j}} =
z_j^{-r^i_{1j}}$.

Thus, if $e^{b_1^j}, e^{\alpha_i/d_0}$, and $e^{\beta_i/d_0}$ are
rational for all $i \in \{1, \dots, p\},\\ {j \in \{1, \dots, N\}}$, one
can show that $\nu(y, u)$ is an $\mathcal{R}$-formula by multiplying
all $h^i_1$ constraints by their denominators. All denominators are
positive since $y_i$ and $z_i$ are constrained to be positive.
\end{proof}

The single-hidden-layer assumption in Theorem~\ref{thm:decidable} is
not too restrictive since DNNs with one hidden layer are still
universal approximators. At the same time, the technique used to prove
Theorem~\ref{thm:decidable} cannot be applied to multiple hidden
layers since the DNN becomes an $\mathcal{R}_{\text{exp}}$-formula in
that case. Note that it might be possible to show more general
versions of Theorem~\ref{thm:decidable} by relaxing the interval
constraints or the real-arithmetic constraints.  Finally, note that
the assumption on the DNN's weights is mild since a DNN's weights can
be altered in such a way that they are arbitrarily close to the
original weights while also satisfying the theorem's requirements.

\section{DNN Reachability Using Hybrid Systems}
\label{sec:approach}
Having analyzed the decidability of DNN reachability in
Section~\ref{sec:decidability}, in this section we investigate an
approach to computing the DNN's reachable set. In particular, we
transform the DNN into an equivalent hybrid system, which allows us to
use existing hybrid system reachability tools such as
Flow*. Sections~\ref{sec:sigmoid_de} and~\ref{sec:dnn2hs} explain the
transformation technique, and Section~\ref{sec:toy_example} provides
an illustrative example. Finally, Section~\ref{sec:tools} discusses
existing hybrid system reachability tools. Note that this section
focuses on the case of sigmoid activations; the treatment of tanh
activations is almost identical -- the differences are noted in the
relevant places in the section.

%% Before presenting the hybrid-system-based technique, we note that
%% other approaches could be developed as well such as the SMT- and
%% MILP-based methods used for ReLU-based DNNs. These approaches were
%% developed for piece-wise linear functions and could be extended to
%% sigmoids by bounding the sigmoid from above and below by piece-wise
%% linear functions. However, there are two main drawbacks associated
%% with these techniques as compared with Verisig: 1) they cannot be
%% composed with the plant's hybrid system in order to verify properties
%% about the closed-loop system; 2) approximating the sigmoid
%% sufficiently well requires multiple linear pieces (around 100 pieces
%% in the case studies presented in this paper), which eventually affects
%% the scalability of these techniques as discussed in
%% Section~\ref{sec:comparison}. For these reasons, we adopt the
%% hybrid-system-based approach that allows to both verify properties
%% about the closed-loop system and achieve very good approximation of
%% the sigmoid function.
%% \JW{This paragraph fits poorly here -- it is almost like a "related
%% work" section ... can it be moved to the introduction?}

\subsection{Sigmoids as solutions to differential equations}
\label{sec:sigmoid_de}
The main observation that allows us to transform a DNN into an
equivalent hybrid system is the fact that the sigmoid derivative can
be expressed in terms of the sigmoid itself:\footnote{The
  corresponding differential equation for tanh is
  $(d\text{tanh}/dx)(x) = 1 - \text{tanh}^2(x)$.}
\begin{align}
\frac{d\sigma}{dx}(x) = \sigma(x) (1 - \sigma(x)).
\end{align}
Thus, the sigmoid can be treated as a quadratic dynamical
system. Since we would like to know the possible values of the sigmoid
for a given set of inputs, we introduce a ``time'' variable $t$ that
is multiplied by the inputs. In particular, consider the proxy
function
\begin{align}
\label{eq:proxy}
g(t,x) = \sigma(tx) = \frac{1}{1 + e^{-xt}},
\end{align}
such that $g(1,x) = \sigma(x)$ and, by the chain rule,
\begin{align}
  \frac{\partial g}{\partial t}(t,x) = \dot{g}(t,x) = x g(t,x) (1 - g(t,x)).
\end{align}
Thus, by tracing the dynamics of $g$ until time $t = 1$, we obtain
exactly the value of $\sigma(x)$; the initial condition is $g(0,x) =
0.5$, as can be verified from~\eqref{eq:proxy}. Since each neuron in a
sigmoid-based DNN is a sigmoid function, we can use the proxy function
$g$ to transform the entire DNN into a hybrid system, as described
next.

\begin{figure*}[!th]
  \centering 
       \subfloat[Example DNN.]{
           \includegraphics[width=0.35\linewidth]{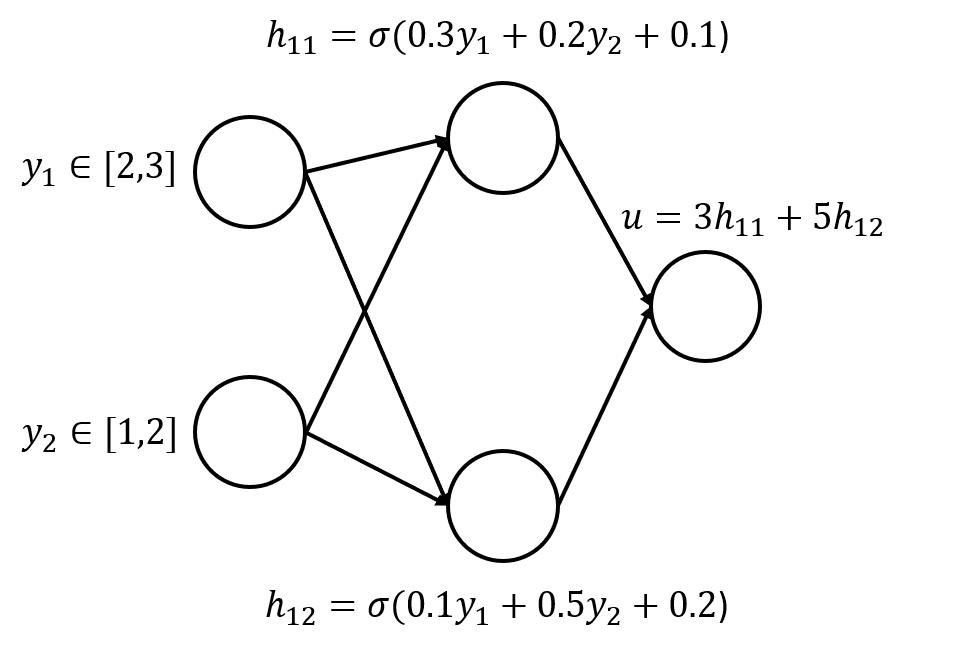}
           \label{fig:dnn_toy}
       }
       \subfloat[Equivalent hybrid system.]{
	   \includegraphics[width=0.64\linewidth]{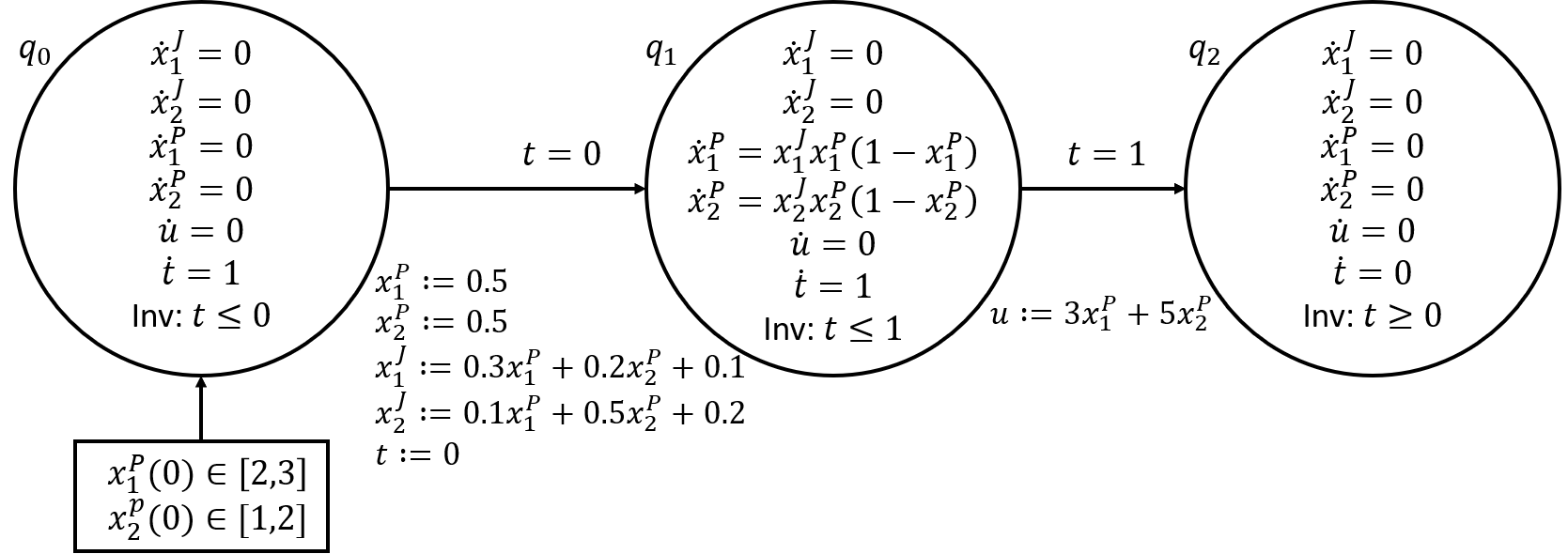}
           \label{fig:hs_toy}             
       }
  \caption{Small example illustrating the transformation from a DNN to a
    hybrid system.}
  \label{fig:toy}
\end{figure*}

\subsection{Deep Neural Networks as Hybrid Systems}
\label{sec:dnn2hs}
Given the proxy function $g$ described in
Section~\ref{sec:sigmoid_de}, we now show how to transform a DNN into
a hybrid system. Let $N_i$ be the number of neurons in hidden layer
$h_i$ and let $h_{ij}$ denote neuron $j$ in $h_i$, i.e.,
\begin{align}
  h_{ij}(x) = \sigma((w_{i}^j)^{\top}x + b_{i}^j),
\end{align}
where $(w_{i}^j)^{\top}$ is row $j$ of $W_i$ and $b_{i}^j$ is element
$j$ of $b_i$. Given $h_{ij}$, the corresponding proxy function
$g_{ij}$ is defined as follows:
\begin{align*}
g_{ij}(t, x) = \sigma(t\cdot ((w_{i}^j)^{\top}x + b_{i}^j)) = \frac{1}{1 + \text{exp}\{-t\cdot((w_{i}^j)^{\top}x + b_{i}^j)\}},
\end{align*}
where, once again, $g_{ij}(1,x) = h_{ij}(x)$. Note that, by the chain
rule,
\begin{align}
\label{eq:proxy_c}
\frac{\partial g_{ij}}{\partial t}(t,x) = \dot{g}_{ij}(t,x) =
((w_{i}^j)^{\top}x + b_{i}^j) g_{ij}(t,x) (1 - g_{ij}(t,x)).
\end{align}
Thus, for a given $x$, the value of hidden layer $h_i(x)$ can be
obtained by tracing all $g_{ij}(t,x)$ until $t=1$ (initialized at
$g_{ij}(0,x) = 0.5$). This suggests that each hidden layer can be
represented as a set of differential equations $\dot{g}_{ij}(t,x)$,
where $g_{ij}$ can be considered a state.

With the above intuition in mind, we now show how to transform the DNN
into an equivalent hybrid system.
%% Recall from Section~\ref{sec:problem} that the
%% safety property we are interested in verifying is $$\phi(X, Y) :=
%% (\psi_1(X) \wedge f(X) = Y) \Rightarrow \psi_2(Y),$$ where $X \in
%% \mathbb{R} ^N$ are the DNN inputs, $Y \in \mathbb{R}$ is the output,
%% and $\psi_1$ and $\psi_2$ are linear constraints.
To simplify notation, we assume $N = N_i$ for all $i \in \{1, \dots,
L-1\}$; we also assume the DNN has only one output. The proposed
approach can be extended to the more general case by adding more
states in the hybrid system.

The hybrid system has one mode for each DNN layer. To ensure the
hybrid system is equivalent to the DNN, in each mode we trace
$g_{ij}(t,x)$ until $t=1$ by using the differential equations
$\dot{g}_{ij}(t,x)$ in~\eqref{eq:proxy_c}. Thus, we use $N$ continuous
states, $x_1^P, \dots, x_N^P$, to represent the proxy variables for
each layer; when in mode $i$, each $x_j^P$, ${j \in \{1, \dots, N\}}$,
represents neuron $h_{ij}$ in the DNN. We also introduce $N$
additional continuous states (one per neuron), $x_1^L, \dots, x_N^L$,
to keep track of the linear functions within each neuron. The $x_i^L$
states are necessary because the inputs to each neuron are functions
of the $x_i^P$ states reached in the previous mode.

The hybrid system description is formalized in
Proposition~\ref{prop:dnn2hs}. The extra mode $q_0$ is used to reset
the $x_i^P$ states to 0.5 and the $x_i^J$ states to their
corresponding values in $q_1$. The two extra states, $t$ and $u$, are
used to store the ``time'' and the DNN's output, respectively. Note
that $\odot$ denotes Hadamard (element-wise) product.
\begin{proposition}
\label{prop:dnn2hs}
Let $h: \mathbb{R}^p \rightarrow \mathbb{R}^1$ be a sigmoid-based DNN
with $L-1$ hidden layers (with $N$ neurons each) and a linear last
layer with one output. The image under $h$ of a given set $I_y$ is
exactly the reachable set for $u$ in mode $q_L$ of the following
hybrid system:
\begin{itemize}

\item Continuous states: $x^P = [x_1^P, \dots, x_N^P]^{\top}, x^J =
  [x_1^J, \dots, x_N^J]^{\top}$, $u, t$;

\item Discrete states (modes): $q_0, q_1, \dots, q_L$;
  
\item Initial states: $x^P \in I_y$, $x^J = 0, u = 0, t = 0$;

\item Flow:

\begin{itemize}
\item $F(q_0) = [\dot{x}^P = 0, \dot{x}^J = 0, \dot{u} = 0, \dot{t} =
  1]$;

\item $F(q_i) = [\dot{x}^P = x^J\odot x^P\odot (1 - x^P), \dot{x}^J =
  0, \dot{u} = 0, \dot{t} = 1]$ for $i \in \{1, \dots, L-1\}$;

\item $F(q_L) = [\dot{x}^P = 0, \dot{x}^J = 0, \dot{u} = 0, \dot{t} =
  0]$;
\end{itemize}

\item Transitions: $E = \{(q_0, q_1), \dots, (q_{L-1}, q_L)\}$;

\item Invariants:

\begin{itemize}

\item $I(q_0) = \{t \le 0\}$;

\item $I(q_i) = \{t \le 1\}$ for $i \in \{1, \dots, L-1\}$;

\item $I(q_L) = \{t \le 0\}$;

\end{itemize}

\item Guards:

\begin{itemize}

\item $G(q_0, q_1) = \{t = 0\}$;

\item $G(q_i, q_{i+1}) = \{t = 1\}$ for $i \in \{1, \dots, L-1\}$;

\end{itemize}

\item Resets:

\begin{itemize}

\item $R(q_i, q_{i+1}) = \{x^P = 0.5, x^J = W^ix^P
  + b^{i}, t = 0\}\\$ for $i \in \{0, \dots, L-2\}$;

\item $R(q_{L-1}, q_L) = \{u = W^L x^P + b^L\}$.
  
\end{itemize}

\end{itemize}
\end{proposition}

\begin{proof}
First note that the reachable set of $x^P$ in mode $q_1$ at time $t =
1$ is exactly the image of $I_y$ under $h_1$, the first hidden
layer. This is true because at $t = 1$, $x^P$ takes the value of the
sigmoid function. Applying this argument inductively, the reachable
set of $x^P$ in mode $q_{L-1}$ at time $t = 1$ is exactly the image of
$I_y$ under $h_{L-1}\circ \dots \circ h_1$. Finally, $u$ is a linear
function of $x^P$ with the same parameters as the last linear layer of
$h$. Thus, the reachable set for $u$ in mode $q_L$ is the image of
$I_y$ under $h_{L}\circ \dots \circ h_1 = h$.
\end{proof}

\subsection{Illustrative Example}
\label{sec:toy_example}
To illustrate the transformation process from a DNN to a hybrid
system, this subsection presents a small example, shown in
Figure~\ref{fig:toy}. The two-layer DNN is transformed into an
equivalent three-mode hybrid system. Since all the weights are
positive and the sigmoids are monotonically increasing, the maximum
value for the DNN's output $u$ is achieved at the maximum values of
the inputs, whereas the minimum value for $u$ is achieved at the
minimum values of the inputs, i.e., $u \ge 3\sigma(0.3\cdot2 +
0.2\cdot1 + 0.1) + 5\sigma(0.1\cdot2 + 0.5\cdot1 + 0.2)$ and $u \le
3\sigma(0.3\cdot3 + 0.2\cdot2 + 0.1) + 5\sigma(0.1\cdot3 + 0.5\cdot2 +
0.2)$. The same conclusion can be reached about state $u$ in the
hybrid system.

%!TEX root = main.tex

\subsection{Hybrid System Verification Tools}
\label{sec:tools}
%% Having described how to transform a DNN into a hybrid system, we now
%% discuss existing hybrid system verification approaches and tools that
%% can be used to solve the verification problem.
Depending on the hybrid system model and the desired precision, there
are multiple tools one might use. In the case of linear hybrid
systems, there are powerful tools that scale up to a few thousand
states~\cite{frehse11}. For non-linear systems, reachability is
undecidable in general, except for specific
subclasses~\cite{alur95,lafferriere99}. Despite this negative result,
multiple reachability methods have been developed that have proven
useful in specific scenarios. In particular, Flow*~\cite{chen13} works
by constructing flowpipe overapproximations of the dynamics in each
mode using Taylor Models; although Flow* provides no decidability
claims, it scales well in practical applications. Alternatively,
dReach~\cite{kong15} provides $\delta$-decidability guarantees for
dynamics described by Type 2 computable functions; at the same time,
dReach is not as scalable and could handle more than a few dozen
variables in the examples tried in this paper. Finally, one can also
use SMT solvers such as z3~\cite{moura08}; yet, SMT solvers are not
optimized for non-linear arithmetic and do not scale well either.

In this paper, we use Flow* due to its scalability; as shown in the
evaluation, it efficiently handles systems with a few hundred states,
i.e., DNNs with a few hundred neurons per layer. Furthermore, the
mildly non-linear nature of the sigmoid dynamics suggests that the
approximations used in Flow* are sufficiently precise so as to verify
interesting properties. This is illustrated in the case studies as
well as in the scalability evaluation in Section~\ref{sec:comparison}.
%% dReach was used because it provides $\delta$-reachability
%% guarantees, i.e., it can be used to verify reachability within any
%% desired precsion. While dReach is not as scalable as Flow*, it can
%% still handle a few dozen variables and can thus be useful for
%% small-size DNNs.

Finally, note that all existing tools have been developed for large
classes of hybrid systems and do not exploit the specific properties
of the sigmoid dynamics, e.g., they are monotonic and polynomial. For
example, in some cases it is possible to symbolically compute the
reachable set of monotone systems~\cite{coogan15}, although directly
applying this approach to our setting does not work due to the large
state space. Thus, developing a specialized sigmoid reachability tool
is bound to greatly improve scalability and precision; since this
paper is a proof of concept, developing such a tool is left for future
work.
%% \RI{Might be good to mention the importance of small weights
%%   somewhere but not sure where is a good place.}  \JW{I am not sure
%%   it adds anything to this paper -- that is an implementation
%%   detail that is used to speed up computation ... perhaps we should
%%   leave that insight for the tool paper.}

%!TEX root = main.tex

\section{Case Study Applications}
\label{sec:results}
This section presents two case studies in order to illustrate possible
use cases for the proposed verification approach. These case studies
were chosen in domains where DNNs are used extensively as controllers,
with weak worst-case guarantees about the trained network. This means
it is essential to verify properties about these closed-loop systems
in order to assure their functionality. The first case study,
presented in Section~\ref{sec:rl}, is Mountain Car, a benchmark
problem in RL. Section~\ref{sec:mpc_dnn} presents the second case
study in which a DNN is used to approximate an MPC with safety
guarantees.

\begin{figure}[!t]
  \centering 
  \includegraphics[width=0.7\linewidth]{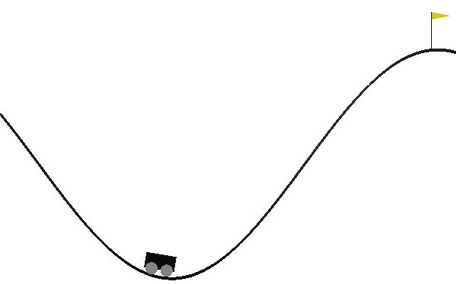}
  \caption{Mountain Car problem~\cite{openaigym}. The car needs to
    drive up the left hill first in order to gather enough momentum
    and reach its goal on the right.}
  \label{fig:Mcar}
\end{figure}

\subsection{Mountain Car: A Reinforcement Learning Case Study}
\label{sec:rl}
This subsection illustrates how Verisig could be used to verify
properties on a benchmark RL problem, namely Mountain Car (MC). In
(MC), an under-powered car must drive up a steep hill, as shown in
Figure~\ref{fig:Mcar}. Since the car does not have enough power to
simply accelerate up the hill, it needs to drive up the opposite hill
first in order to gather enough momentum to reach its goal. The
learning task is to learn a controller that takes as input the car's
position and velocity and outputs an acceleration command. The car has
the following discrete-time dynamics:
\begin{align*}
p_{k+1} &= p_k + v_{k}\\
v_{k+1} &= v_k + 0.0015u_k - 0.0025*cos(3p_k),
\end{align*}
where $u_k$ is the controller's input, and $p_k$ and $v_k$ are the
car's position and velocity, respectively, with $p_0$ chosen uniformly
at random from $[-0.6, -0.4]$ and $v_0 = 0$. Note that $v_k$ is
constrained to be within $[-0.07, 0.07]$ and $p_k$ is constrained to
be within $[-1.2, 0.6]$, thereby introducing (hybrid) mode switches
when these constraints are violated. We consider the continuous
version of the problem such that $u_k$ is a real number between -1 and
1.

During training, the learning algorithm tries different control
actions and observes a reward.  The reward associated with a control
action $u_k$ is $-0.1u_k^2$, i.e., larger control inputs are penalized
more so as to avoid a ``bang-bang'' strategy. A reward of 100 is
received when the car reaches its goal. The goal of the training
algorithm is to maximize the car's reward. The training stage
typically occurs over multiple episodes (if not solved, an episode is
terminated after 1000 steps) such that various behaviors can be
observed.  MC is considered ``solved'' if, during testing, the car
goes up the hill with an average reward of at least 90 over 100
consecutive trials.

Using Verisig, one can strengthen the definition of a ``solved'' task
and \emph{verify} that the car will go up the hill with a reward of at
least 90 starting from \emph{any} initial condition. To illustrate
this, we trained a DNN controller for MC in OpenAI
Gym~\cite{openaigym}, a toolkit for developing and comparing
algorithms on benchmark RL problems. We utilized a standard
actor/critic approach for deep RL problems~\cite{lillicrap15}. This is
a two-DNN setting in which one DNN (the critic) learns the reward
function, whereas the other one (the actor) learns the control. Once
training is finished, the actor is deployed as the DNN controller for
the closed-loop system. We trained a two-hidden-layer sigmoid-based
DNN with 16 neurons per layer; the last layer has a tanh activation
function in order to scale the output to be between -1 and 1. Note
that larger networks were also trained in order to evaluate
scalability, as discussed in Section~\ref{sec:comparison}.

To verify that the car will go up the hill with a reward of at least
90, we transform the DNN into an equivalent hybrid system using
Verisig and compose it with the car's hybrid system. We use
Verisig+Flow* to verify the desired property on the composed system,
given any initial position in [-0.6, -0.4]. Note that we split the
initial condition into subsets and verify the property for each subset
separately. This is necessary because the DNN takes very different
actions from different initial conditions, e.g., large negative inputs
when the car is started from the leftmost position and small negative
inputs for larger initial conditions. This variability introduces
uncertainty in the dynamics and causes large approximation errors.

\begin{table}
\centering
\begin{tabular}{|l|c|c|c|c|}
\hline
Initial condition   & Verified  & Reward & \# steps & Time \\ \hline
[-0.41, -0.40] & Yes & >= 90  & <= 100  & 1336s  \\ \hline
[-0.415, -0.41] & Yes & >= 90  & <= 100  & 1424s  \\ \hline
[-0.42, -0.415] & Yes & >= 90  & <= 100  & 812s  \\ \hline
[-0.43, -0.42] & Yes & >= 90  & <= 100  & 852s  \\ \hline
[-0.45, -0.43] & Yes & >= 90  & <= 100  & 886s  \\ \hline
[-0.48, -0.45] & Yes & >= 90  & <= 100  & 744s  \\ \hline
[-0.50, -0.48] & Yes & >= 90  & <= 100  & 465s  \\ \hline
[-0.53, -0.50] & Yes & >= 90  & <= 100  & 694s  \\ \hline
[-0.55, -0.53] & Yes & >= 90  & <= 100  & 670s  \\ \hline
[-0.57, -0.55] & Yes & >= 90  & <= 100  & 763s  \\ \hline
[-0.58, -0.57] & Yes & >= 90  & <= 109  & 793s  \\ \hline
[-0.59, -0.58] & Yes & >= 90  & <= 112  & 1307s  \\ \hline
[-0.6, -0.59] & No & N/A  & N/A  & N/A  \\ \hline
\end{tabular}
\caption{Verisig+Flow* verification times (in seconds) for different
  initial conditions of MC. The third column shows the verified lower
  bound of reward. The fourth column shows the verified upper bound of
  the number of dynamics steps.}
\label{tab:mc}
\end{table}

\begin{figure}[!t]
  \centering 
  \includegraphics[width=0.9\linewidth]{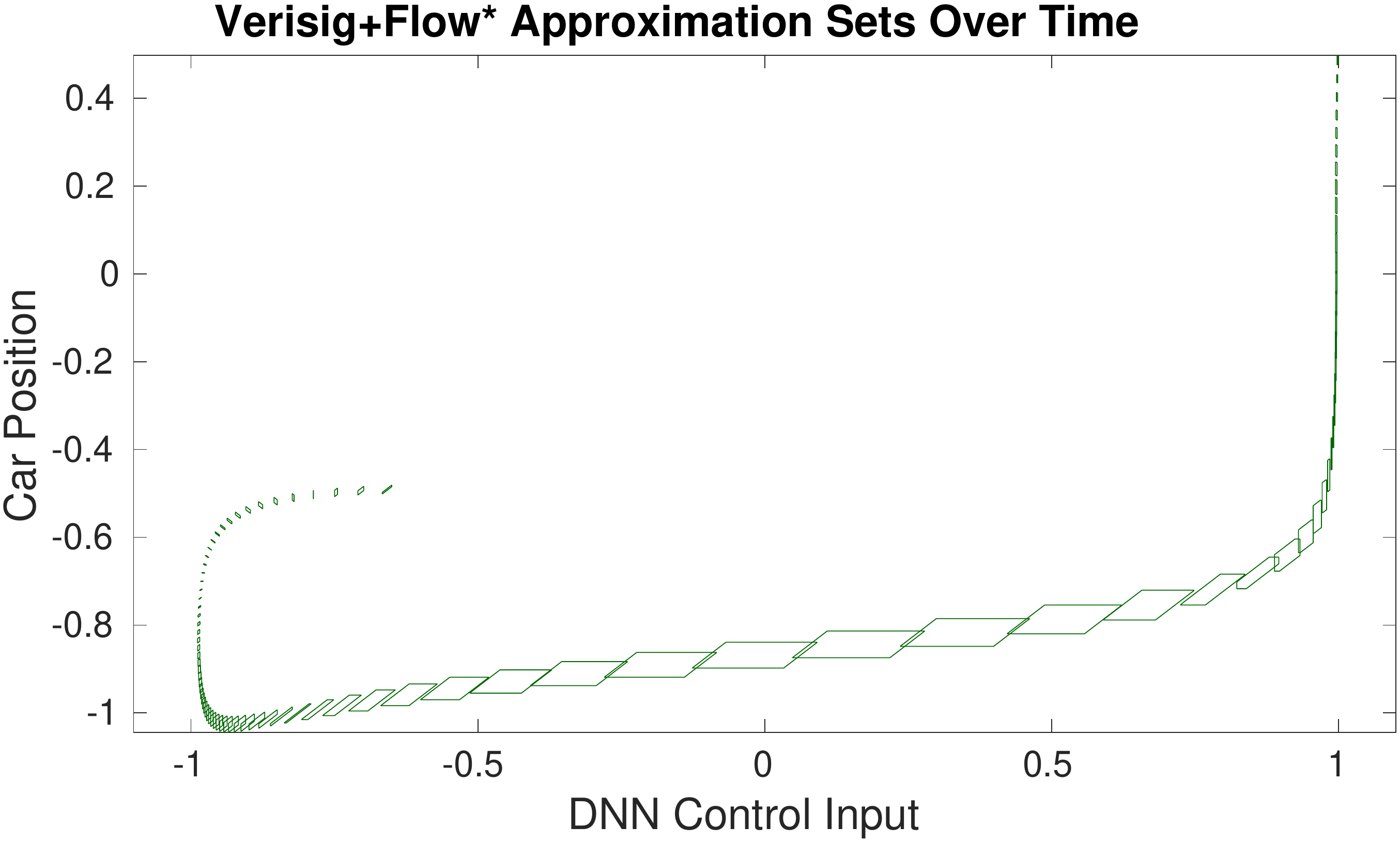}
  \caption{Verisig+Flow* approximation sets over time.}
  \label{fig:uncertainty}
\end{figure}

Table~\ref{tab:mc} presents the verification times for each
subset. Most properties are verified within 10-15 minutes; the
properties at either end of the initial set take longer to verify due
to branching in the car's hybrid system as caused by the car reaching
the minimum allowed position. For most initial conditions, we verify
that the car will go up the hill with a reward of at least 90 and in
at most 100 dynamics steps. Interestingly, after failing to verify the
property for the subset [-0.6, -0.59], we found a counter-example when
starting the car from $p_0 = -0.6$: the final reward was 88. This
suggests that Verisig is not only useful for verifying properties of
interest but it can also be used to identify areas for which these
properties do not hold. In the case of MC, this information can be
used to retrain the DNN by starting more episodes from [-0.6, -0.59]
since the likely reason the DNN does not perform well from that
initial set is that not many episodes were started from there during
training.

Finally, we illustrate the progression of the approximation sets
created by Flow*. Figure~\ref{fig:uncertainty} shows a two-dimensional
projection of the approximation sets over time (for the case $p_0 \in
[-0.5, -0.48]$), with the DNN control inputs plotted on the x-axis and
the car's position on the y-axis. Initially, the uncertainty is fairly
small and remains so until the car goes up the left hill and starts
going quickly downhill. At that point, the uncertainty increases but
it remains within the tolerance necessary to verify the desired
property.

\subsection{Using DNNs to Approximate MPCs with Safety Guarantees}
\label{sec:mpc_dnn}
To further evaluate the applicability of Verisig, we also consider a
case study in which a DNN is used to approximate an MPC with safety
guarantees. DNNs are used to approximate controllers for several
reasons: 1) the MPC computation is not feasible at
run-time~\cite{hertneck18}; 2) storing the original controller (e.g.,
as a lookup table) requires too much memory~\cite{julian16}; 3)
performing reachability analysis by discretizing the state space is
infeasible for high-dimensional systems ~\cite{royo18}. We focus on
the latter scenario in which the aim is to develop a DNN controller
with safety guarantees.

\begin{figure}[!t]
  \centering 
  \includegraphics[width=0.9\linewidth]{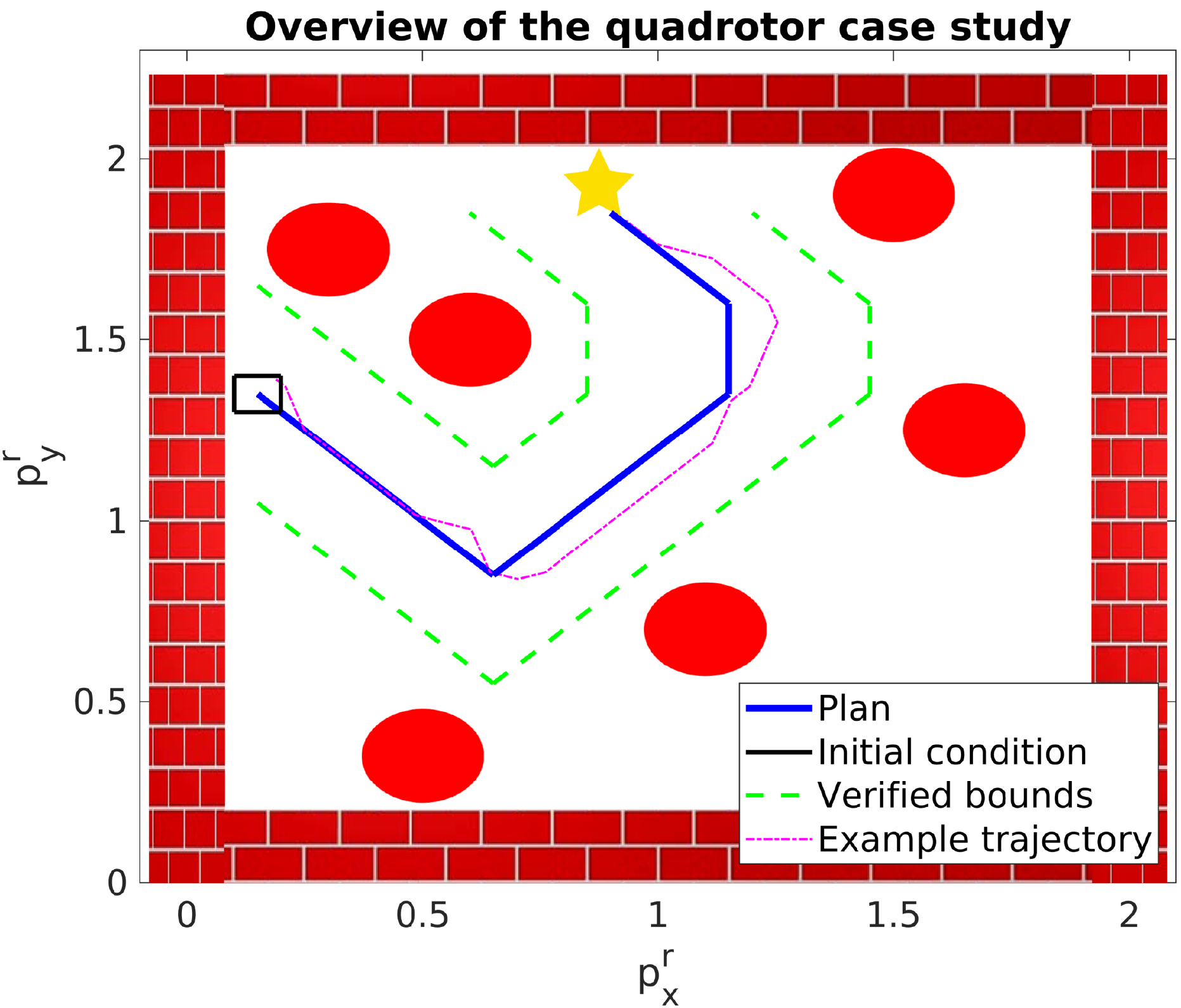}
  \caption{Overview of the quadrotor case study, as projected to the
    $(x,y)$-plane. The quadrotor's goal is to follow its planner in
    order to reach the goal at the top (star) without colliding into
    obstacles (red circles). We have verified that, starting from any
    state within the shown set of initial conditions, the quadrotor
    will not deviate from its plan by more than 0.32m (and hence will
    not collide into any obstacles).}
  \label{fig:quad_cs}
\end{figure}

As described in prior work~\cite{royo18}, it is possible to train a
DNN approximating an MPC in the case of control-affine systems whose
goal is to follow a piecewise-linear plan. In this case, the optimal
controller is ``bang-bang'', i.e., it is effectively a classifier
mapping a system state to one of finitely many control actions. Given
a trained DNN, one can simulate the closed-loop system over a horizon
$T$ with a worst-case (i.e., most difficult to follow) plan -- the
largest deviation from this plan (which also follows a ``bang-bang''
strategy) is a worst-case guarantee for the deviation from \emph{any}
other plan over horizon $T$. Thus, we obtain safety guarantees for the
system assuming that it is always started from the same initial
condition.

In this case study, we consider a six-dimensional control-affine model
for a quadrotor controlled by a DNN and verify that the quadrotor, as
started from a set of initial conditions, will reach its goal without
colliding into nearby obstacles. Specifically, the quadrotor follows a
path planner, given as a piecewise-linear system, and tries to stay as
close to the planner as possible. The setup, as projected to the
$(x,y)$-plane, is shown in Figure~\ref{fig:quad_cs}. The quadrotor and
planner dynamics models are as follows:
\begin{align}
\label{eq:quad_model}
\dot{q} := \left[\begin{array}{c}
\dot{p}_x^q\\
\dot{p}_y^q\\
\dot{p}_z^q\\
\dot{v}_x^q\\
\dot{v}_y^q\\
\dot{v}_z^q\end{array}\right] = \left[\begin{array}{c}
v_x^q\\
v_y^q\\
v_z^q\\
g\text{tan}\theta\\
-g\text{tan}\phi\\
\tau - g\end{array}\right], \dot{p} := \left[\begin{array}{c}
\dot{p}_x^p\\
\dot{p}_y^p\\
\dot{p}_z^p\\
\dot{v}_x^p\\
\dot{v}_y^p\\
\dot{v}_z^p\end{array}\right] = \left[\begin{array}{c}
b_x\\
b_y\\
b_z\\
0\\
0\\
0\end{array}\right],
\end{align}
where $p_x^q, p_y^q, p_z^q$ and $p_x^p, p_y^p, p_z^p$ are the
quadrotor and planner's positions, respectively; $v_x^q, v_y^q, v_z^q$
and $v_x^p, v_y^p, v_z^p$ are the quadrotor and planner's velocities,
respectively; $\theta$, $\phi$ and $\tau$ are control inputs (for
pitch, roll and thrust); $g = 9.81 m/s^2$ is gravity; $b_x, b_y, b_z$
are piecewise constant. The control inputs have constraints $\phi,
\theta \in [-0.1, 0.1]$ and $\tau \in [7.81, 11.81]$; the planner
velocities have constraints $b_x, b_y, b_z \in [-0.25, 0.25]$. The
controller's goal is to ensure the quadrotor is as close to the
planner as possible, i.e., stabilize the system of relative states $r
:= [p_x^r, p_y^r, p_z^r, v_x^r, v_y^r, v_z^r]^{\top} = q - p$.

\begin{table}
\centering
\begin{tabular}{|l|c|c|}
\hline
Initial condition on $(p_x^r, p_y^r)$  & Property & Time \\ \hline
$[-0.05, -0.025] \times [-0.05, -0.025]$ & $\|r_3\|_{\infty} \le 0.32m$ & 2766s  \\ \hline
$[-0.025, 0] \times [-0.05, -0.025]$ & $\|r_3\|_{\infty} \le 0.32m$ & 2136s  \\ \hline
$[0, 0.025] \times [-0.05, -0.025]$ & $\|r_3\|_{\infty} \le 0.32m$ & 2515s  \\ \hline
$[0.025, 0.05] \times [-0.05, -0.025]$ & $\|r_3\|_{\infty} \le 0.32m$ & 897s  \\ \hline
$[-0.05, -0.025] \times [-0.025, 0]$ & $\|r_3\|_{\infty} \le 0.32m$ & 1837s  \\ \hline
$[-0.025, 0] \times [-0.025, 0]$ & $\|r_3\|_{\infty} \le 0.32m$ & 1127s  \\ \hline
$[0, 0.025] \times [-0.025, 0]$ & $\|r_3\|_{\infty} \le 0.32m$ & 1593s  \\ \hline
$[0.025, 0.05] \times [-0.025, 0]$ & $\|r_3\|_{\infty} \le 0.32m$ & 894s  \\ \hline
$[-0.05, -0.025] \times [0, 0.025]$ & $\|r_3\|_{\infty} \le 0.32m$ & 1376s  \\ \hline
$[-0.025, 0] \times [0, 0.025]$ & $\|r_3\|_{\infty} \le 0.32m$ & 953s  \\ \hline
$[0, 0.025] \times [0, 0.025]$ & $\|r_3\|_{\infty} \le 0.32m$ & 1038s  \\ \hline
$[0.025, 0.05] \times [0, 0.025]$ & $\|r_3\|_{\infty} \le 0.32m$ & 647s  \\ \hline
$[-0.05, -0.025] \times [0.025, 0.05]$ & $\|r_3\|_{\infty} \le 0.32m$ & 3534s  \\ \hline
$[-0.025, 0] \times [0.025, 0.05]$ & $\|r_3\|_{\infty} \le 0.32m$ & 2491s  \\ \hline
$[0, 0.025] \times [0.025, 0.05]$ & $\|r_3\|_{\infty} \le 0.32m$ & 2142s  \\ \hline
$[0.025, 0.05] \times [0.025, 0.05]$ & $\|r_3\|_{\infty} \le 0.32m$ & 1090s  \\ \hline
\end{tabular}
\caption{Verisig+Flow* verification times (in seconds) for different
  initial conditions of the quadrotor case study. All properties were
  verified. Note that $r_3 = [p_x^r, p_y^r, p_z^r]$.}
\label{tab:quad_cs}
\vspace{-20pt}
\end{table}

\begin{figure*}[!th]
  \centering 
       \subfloat[16 neurons per layer.]{
           \includegraphics[width=0.24\linewidth]{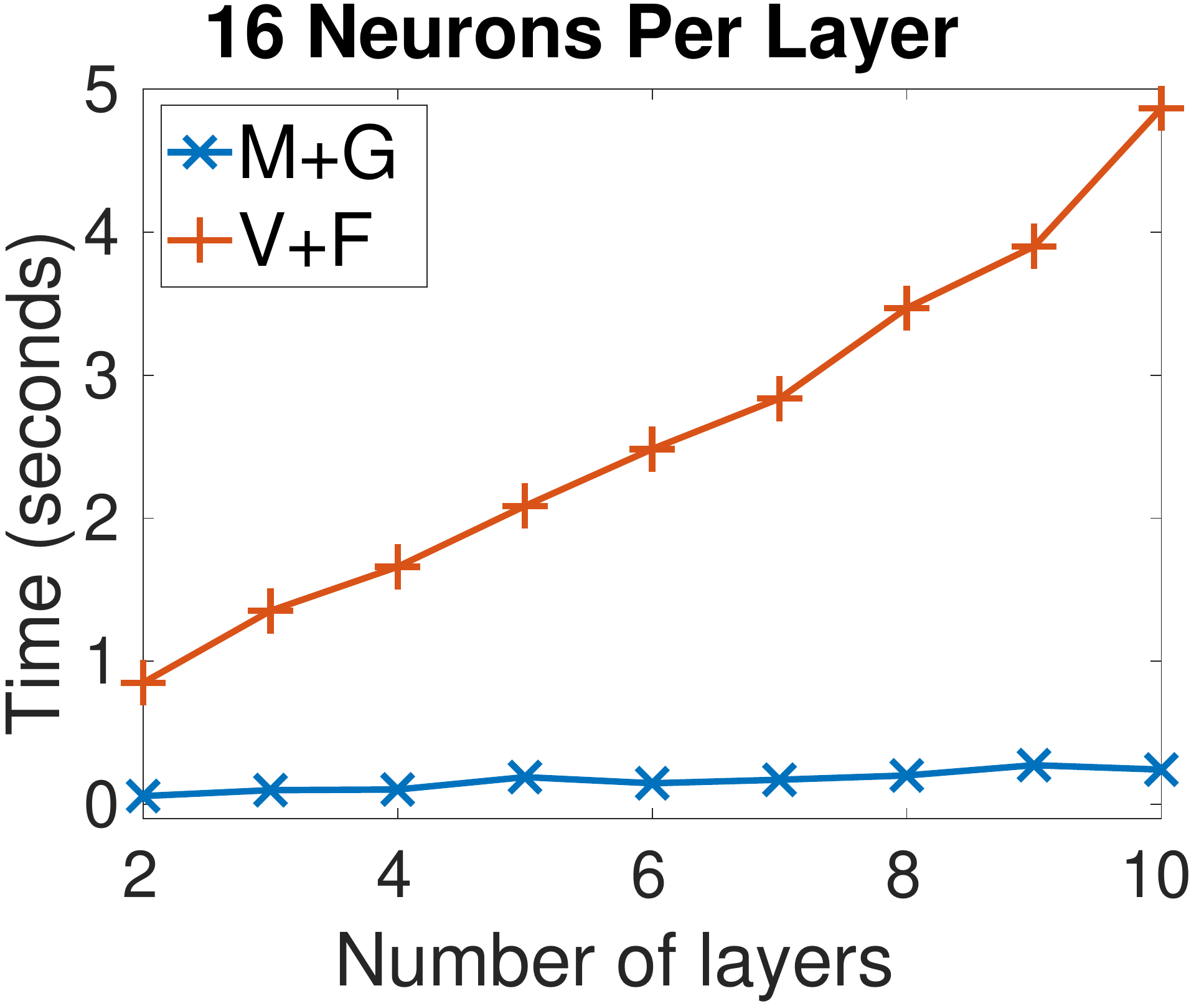}
           \label{fig:16_neurons}
       }
       \subfloat[32 neurons per layer.]{
           \includegraphics[width=0.24\linewidth]{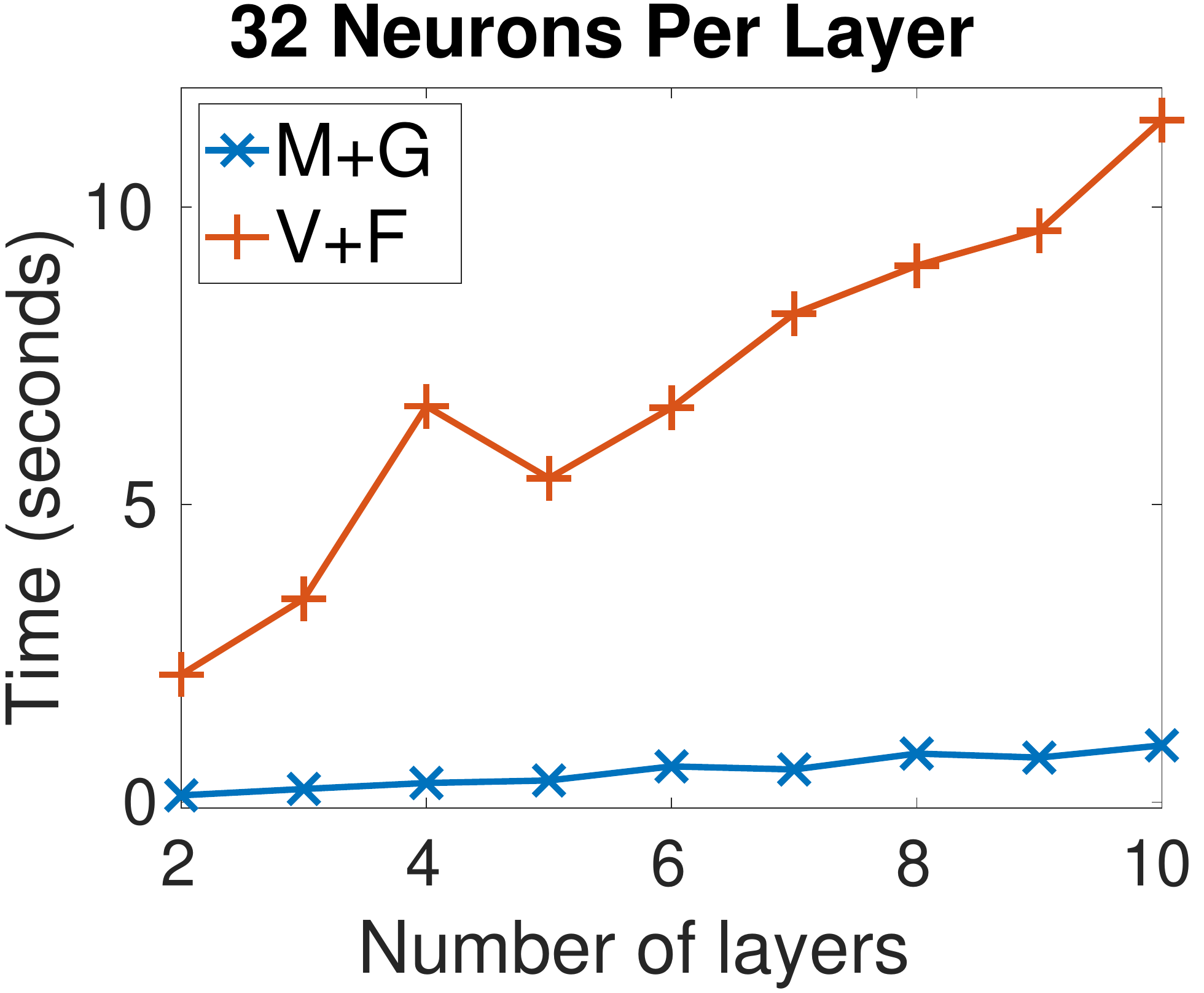}
           \label{fig:32_neurons}
       }
       \subfloat[64 neurons per layer.]{
           \includegraphics[width=0.25\linewidth]{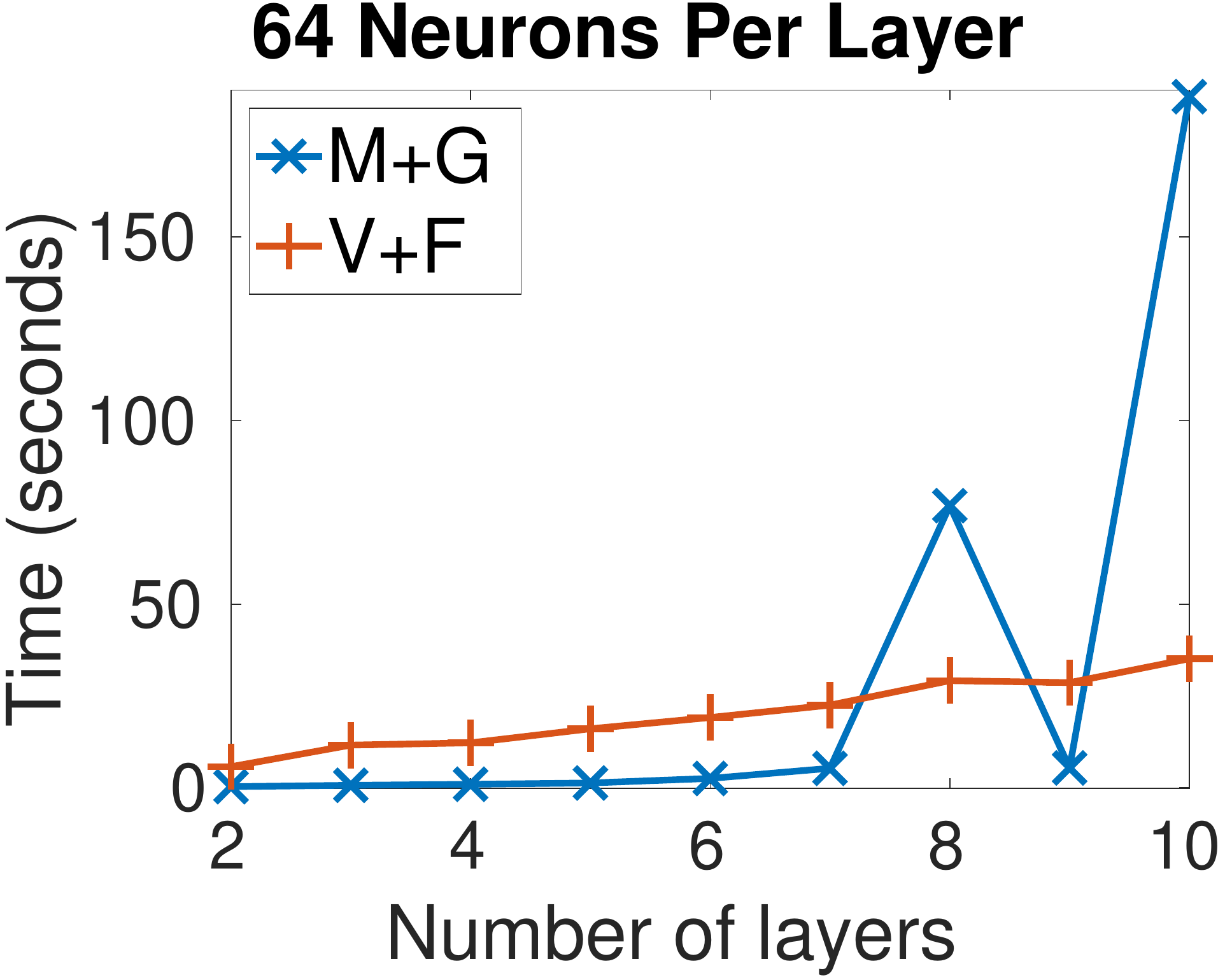}
           \label{fig:64_neurons}
       }
       \subfloat[128 neurons per layer.]{
           \includegraphics[width=0.24\linewidth]{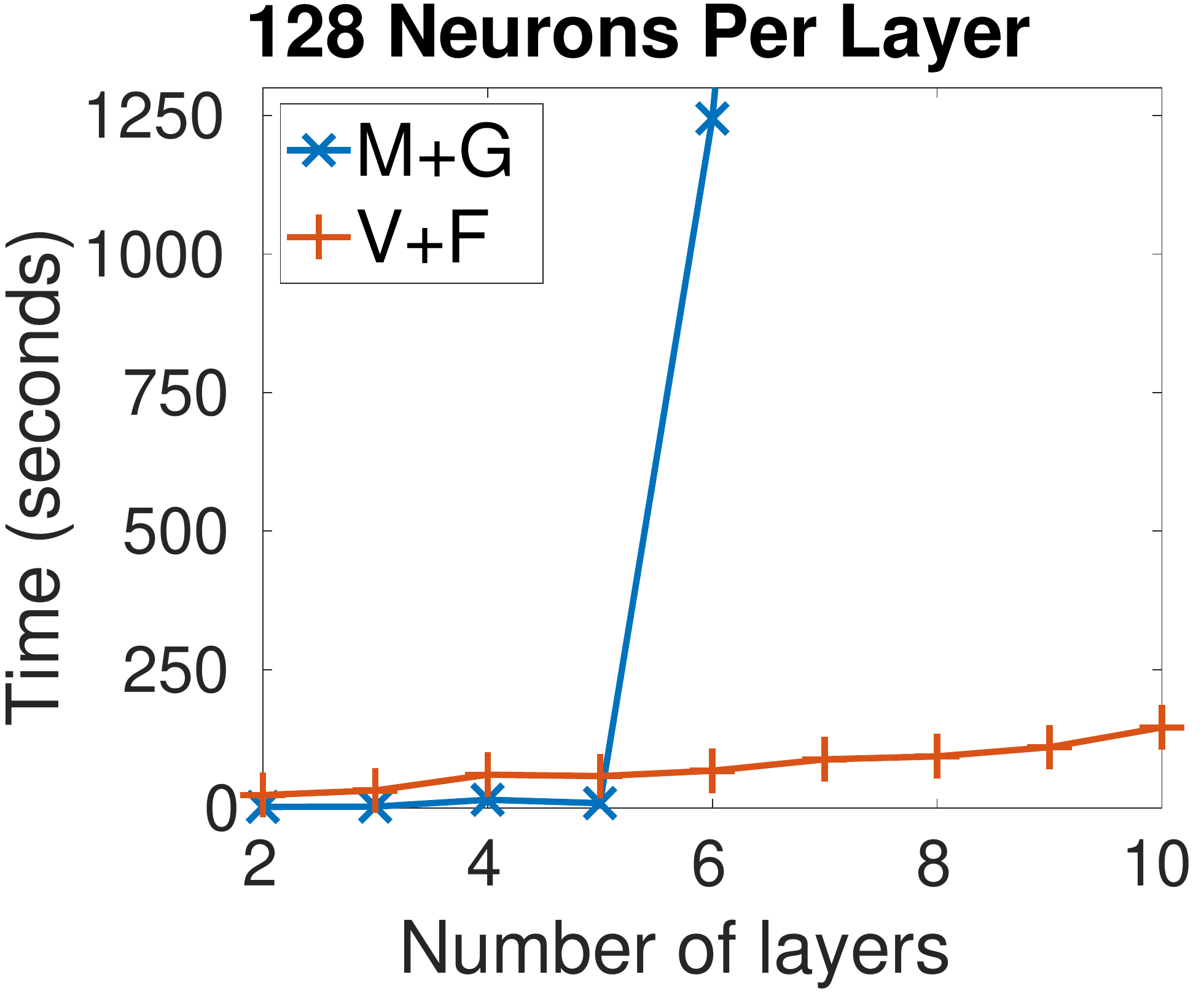}
           \label{fig:128_neurons}
       }       
  \caption{Comparison between the verification times of Verisig+Flow* (V+F) and the MILP-based approach with Gurobi (M+G) for DNNs of increasing size. In each figure, the number of neurons is fixed and number of layers varies from two to 10.}
  \label{fig:comparison}
\vspace{-10pt}
\end{figure*}

To train a DNN controller for the model in~\eqref{eq:quad_model}, we
follow the approach described in prior work~\cite{royo18}. We sample
multiple points from the state space over a horizon $T$ and train a
sequence of DNNs, one for each dynamics step (as discretized using the
Runge-Kutta method). Once two consecutive DNNs have similar training
error, we interrupt training and pick the last DNN as the final
controller. The DNN takes a relative state as input and outputs one of
eight possible actions (the ``bang-bang'' strategy implies there are
two options per control action). We trained a two-hidden layer
tanh-based DNN, with 20 neurons per layer and a linear last layer.

Given the trained DNN controller, we verify a safety property for the
the setup whose $(x,y)$-plane projection is shown in
Figure~\ref{fig:quad_cs}. Specifically, the quadrotor is started from
an initial condition $(p_x^r(0), p_y^r(0)) \in [-0.05, -0.05] \times
[-0.05, -0.05]$ (the other states are initialized at 0) and needs to
stay within 0.32m from the planner in order to reach its goal without
colliding into obstacles. Similar to the MC case study, we split the
initial condition into smaller subsets and verify the property for
each subset.

The verification times of Verisig+Flow* for each subset are shown in
Table~\ref{tab:quad_cs}. Most cases take less than 30 minutes to
verify, which is acceptable for an offline computation. Note that this
verification task is harder than MC not because of the larger
dimension of the state space but because of the discrete DNN
outputs. This means that Verisig+Flow* needs to enumerate and verify
all possible paths from the initial set. This process is
computationally expensive since the number of paths could grow
exponentially with the length of the scenario (set to 30 steps in this
case study). One approach to reduce the computation time would be to
use the Markov property of dynamical systems and skip states that have
been verified previously. We plan to explore this idea as part of
future work.

In summary, this section shows that Verisig can verify both safety and
liveness properties in different and challenging domains. The plant
models can be nonlinear systems specified in either discrete or
continuous time. The next section shows that Verisig+Flow* also scales
well to larger DNNs and is competitive with other approaches for
verification of DNN properties in isolation.

%!TEX root = main.tex

\section{Comparison with other DNN verification techniques}
\label{sec:comparison}

This section complements the Verisig evaluation in
Section~\ref{sec:results} by analyzing the scalability of the proposed
approach. We train DNNs of increasing size on the MC problem and
compare the verification times against the times produced by another
suggested approach to the verification of sigmoid-based DNNs, namely
one using a MILP formulation of the problem~\cite{dutta18}. We verify
properties about DNNs only (without considering the closed-loop
system), since existing approaches cannot be used to argue about the
closed-loop system.

As noted in the introduction, the two main classes of DNN verification
techniques that have been developed so far are SMT- and MILP-based
approaches to the verification of ReLU-based DNNs. Since both of these
techniques were developed for piecewise-linear activation functions,
neither of them can be directly applied to sigmoid-based DNNs. Yet, it
is possible to extend them to sigmoids by bounding the sigmoid from
above and below by piecewise-linear functions. In particular, we
implement the MILP-based approach for comparison purposes since it can
also be used to reason about the reachability of a DNN, similar to
Verisig+Flow*.

The encoding of each sigmoid-based neuron into an MILP problem is
described in detail in~\cite{dutta18}. It makes use of the so called
Big M method~\cite{vanderbei15}, where conservative upper and lower
bounds are derived for each neuron using interval analysis. The
encoding uses a binary variable for each linear piece of the
approximating function such that when that variable is equal to 1, the
inputs are within the bounds of that linear piece (all binary
variables have to sum up to 1 in order to enforce that the inputs are
within the bounds of exactly one linear piece). Thus, the MILP
contains as many binary variables per neuron as there are linear
pieces in the approximating function. Finally, one can use Gurobi to
solve the MILP and compute a reachable set of the outputs given
constraints on the inputs.

To compare the scalability of the two approaches, we trained multiple
DNNs on the MC problem by varying the number of layers from two to ten
and the number of neurons per layer from 16 to 128. A DNN is assumed
to be ``trained'' if most tested episodes result in a reward of at
least 90 -- since this is a scalability comparison only, no
closed-loop properties were verified. For each trained DNN, we record
the time to compute the reachable set of control actions for input
constraints $p_0 \in [-0.52, -0.5]$ and $v_0 = 0$ using both
Verisig+Flow* and the MILP-based approach. For fair comparison, the
two techniques were tuned to have similar approximation error; thus,
we used roughly 100 linear pieces to approximate the sigmoid.

The comparison is shown in Figure~\ref{fig:comparison}. The MILP-based
approach is faster for small networks and for large networks with few
layers. As the number of layers is increased, however, the MILP-based
approach's runtimes increase exponentially due to the increasing
number of binary variables in the MILP. Verisig+Flow*, on the other
hand, scales linearly with the number of layers since the same
computation is run for each layer (i.e., in each mode). This means
that Verisig+Flow* can verify properties about fairly deep networks;
this fact is noteworthy since deeper networks have been shown to learn
more efficiently than shallow ones~\cite{rolnick17,telgarsky16}.

Another interesting aspect of the behavior of the MILP-based approach
can be seen in Figure~\ref{fig:64_neurons}. The verification time for
the nine-layer DNN is much faster than for the eight-layer one,
probably due to Gurobi exploiting a corner case in that specific
MILP. This suggests that the fast verification times of the MILP-based
approach should be treated with caution as it is not known which
example can trigger a worst-case behavior. In conclusion,
Verisig+Flow* scales linearly and predictably with the number of
layers and can be used in a wide range of closed-loop systems with DNN
controllers.
%% \JW{You need to explain the anomaly in Figure 5c -- also, I think
%%   the figure would look better with time in a log-scale and putting
%%   all 8 results in the same figure ... that way we present the
%%   results in 1 figure, and that figure shows we win.  Right now I
%%   have to read the caption and scales to get the information that
%%   we won.  I think the take away is that we are slower on small
%%   ones (but not significantly) and we are a LOT faster on big
%%   ones.}

\section{Conclusion and Future Work}
\label{sec:discussion}
This paper presented Verisig, a hybrid system approach to verifying
safety properties of closed-loop systems using sigmoid-based DNNs as
controllers. We showed that the verification problem is decidable for
DNNs with one hidden layer and decidable for general DNNs if
Schanuel's conjecture is true. The proposed technique uses the fact
that the sigmoid is a solution to a quadratic differential equation,
which allows us to transform the DNN into an equivalent hybrid
system. Given this transformation, we cast the DNN verification
problem into a hybrid system verification problem, which can be solved
by existing reachability tools such as Flow*. We evaluated both the
applicability and scalability of Verisig+Flow* using two case studies,
one from reinforcement learning and one where the DNN was used to
approximate an MPC with safety guarantees.

The novelty of the proposed approach suggests multiple avenues for
future work. First of all, it would be interesting to investigate
whether one could use sigmoid-based DNNs to approximate DNNs with
other activation functions (with analytically bounded error). This
would enable us to verify properties about arbitrary DNNs and would
greatly expand the application domain of Verisig.

A second research direction is to exploit the specific properties of
the sigmoid dynamics, namely the fact that they are monotone and
quadratic, in order to speed up the verification computation. Although
the proposed technique is already scalable to a large class of
applications, it still makes use of Flow*, which is a general-purpose
tool that was developed for a large class of hybrid systems. That is
why, developing a specialized sigmoid verification tool might bring
significant benefits in terms of scalability and precision.

Finally, although the coarse approximation used in Flow* might be seen
as a limitation, no experiments we have run have shown large
approximation errors. This suggests that the Flow* approximation may
be well suited for sigmoid dynamics. Exploring this phenomenon further
and bounding the approximation error incurred by Flow* is also an
intriguing direction for future work.

%% The immediate avenue for future work is finding a real-data
%% application for the proposed approach. This will not only confirm the
%% usefulness of our technique but it will also provide evaluation on a
%% trained DNN as opposed to a randomly generated one. This is important
%% since the scalability of Verisig is affected by the magnitude of the
%% weights -- the reachability computation slows dramatically when the
%% sigmoids get very close to either 0 or 1 since higher approximation
%% precision is required. Thus, it is essential to confirm the simulation
%% scalability results using a real-data case study.

%% The two main directions for future work on the technical side are
%% improving the approximation performance and scalability. The latter
%% would require a more in-depth analysis of the approximation error
%% incurred by Flow* and understanding how it might be
%% approved. Scalability can be enhanced by exploiting the structure of
%% the sigmiod dynamics, e.g., the fact that it is piecewise monotonic;
%% such observations might help us optimize Flow* or develop a completely
%% new tool tailored specifically to sigmiods.

\section*{Acknowledgments}
We thank Xin Chen (University of Dayton, Ohio) for his help with
encoding the case studies in Flow*. We also thank Vicen\c c Rubies
Royo (University of California, Berkeley) for sharing and explaining
his code on approximating MPCs with DNNs. Last, but not least, we
thank Luan Nguyen and Oleg Sokolsky (University of Pennsylvania) for
fruitful discussions about the verification technique.

\bibliographystyle{ACM-Reference-Format}
\bibliography{bibFile}

%%% -*-BibTeX-*-
%%% Do NOT edit. File created by BibTeX with style
%%% ACM-Reference-Format-Journals [18-Jan-2012].

\begin{thebibliography}{37}

%%% ====================================================================
%%% NOTE TO THE USER: you can override these defaults by providing
%%% customized versions of any of these macros before the \bibliography
%%% command.  Each of them MUST provide its own final punctuation,
%%% except for \shownote{}, \showDOI{}, and \showURL{}.  The latter two
%%% do not use final punctuation, in order to avoid confusing it with
%%% the Web address.
%%%
%%% To suppress output of a particular field, define its macro to expand
%%% to an empty string, or better, \unskip, like this:
%%%
%%% \newcommand{\showDOI}[1]{\unskip}   % LaTeX syntax
%%%
%%% \def \showDOI #1{\unskip}           % plain TeX syntax
%%%
%%% ====================================================================

\ifx \showCODEN    \undefined \def \showCODEN     #1{\unskip}     \fi
\ifx \showDOI      \undefined \def \showDOI       #1{#1}\fi
\ifx \showISBNx    \undefined \def \showISBNx     #1{\unskip}     \fi
\ifx \showISBNxiii \undefined \def \showISBNxiii  #1{\unskip}     \fi
\ifx \showISSN     \undefined \def \showISSN      #1{\unskip}     \fi
\ifx \showLCCN     \undefined \def \showLCCN      #1{\unskip}     \fi
\ifx \shownote     \undefined \def \shownote      #1{#1}          \fi
\ifx \showarticletitle \undefined \def \showarticletitle #1{#1}   \fi
\ifx \showURL      \undefined \def \showURL       {\relax}        \fi
% The following commands are used for tagged output and should be
% invisible to TeX
\providecommand\bibfield[2]{#2}
\providecommand\bibinfo[2]{#2}
\providecommand\natexlab[1]{#1}
\providecommand\showeprint[2][]{arXiv:#2}

\bibitem[\protect\citeauthoryear{Administration}{Administration}{[n. d.]}]%
        {tesla_report}
\bibfield{author}{\bibinfo{person}{US~National Highway Traffic~Safety
  Administration}.} \bibinfo{year}{[n. d.]}\natexlab{}.
\newblock \bibinfo{title}{Investigation PE 16-007}.
\newblock
\newblock
\newblock
\shownote{https://static.nhtsa.gov/odi/inv/2016/INCLA-PE16007-7876.pdf.}


\bibitem[\protect\citeauthoryear{Alur, Courcoubetis, Halbwachs, Henzinger, Ho,
  Nicollin, Olivero, Sifakis, and Yovine}{Alur et~al\mbox{.}}{1995}]%
        {alur95}
\bibfield{author}{\bibinfo{person}{R. Alur}, \bibinfo{person}{C. Courcoubetis},
  \bibinfo{person}{N. Halbwachs}, \bibinfo{person}{T.~A. Henzinger},
  \bibinfo{person}{P.~H. Ho}, \bibinfo{person}{X. Nicollin},
  \bibinfo{person}{A. Olivero}, \bibinfo{person}{J. Sifakis}, {and}
  \bibinfo{person}{S. Yovine}.} \bibinfo{year}{1995}\natexlab{}.
\newblock \showarticletitle{The algorithmic analysis of hybrid systems}.
\newblock \bibinfo{journal}{\emph{Theoretical computer science}}
  \bibinfo{volume}{138}, \bibinfo{number}{1} (\bibinfo{year}{1995}),
  \bibinfo{pages}{3--34}.
\newblock


\bibitem[\protect\citeauthoryear{Board}{Board}{[n. d.]}]%
        {uber_report}
\bibfield{author}{\bibinfo{person}{US~National Transportation~Safety Board}.}
  \bibinfo{year}{[n. d.]}\natexlab{}.
\newblock \bibinfo{title}{Preliminary Report Highway HWY18MH010}.
\newblock
\newblock
\newblock
\shownote{https://www.ntsb.gov/investi-gations/AccidentReports/Reports/HWY18MH010-prelim.pdf.}


\bibitem[\protect\citeauthoryear{Chen, {\'A}brah{\'a}m, and
  Sankaranarayanan}{Chen et~al\mbox{.}}{2013}]%
        {chen13}
\bibfield{author}{\bibinfo{person}{X. Chen}, \bibinfo{person}{E.
  {\'A}brah{\'a}m}, {and} \bibinfo{person}{S. Sankaranarayanan}.}
  \bibinfo{year}{2013}\natexlab{}.
\newblock \showarticletitle{Flow*: An analyzer for non-linear hybrid systems}.
  In \bibinfo{booktitle}{\emph{International Conference on Computer Aided
  Verification}}. Springer, \bibinfo{pages}{258--263}.
\newblock


\bibitem[\protect\citeauthoryear{Coogan and Arcak}{Coogan and Arcak}{2015}]%
        {coogan15}
\bibfield{author}{\bibinfo{person}{S. Coogan} {and} \bibinfo{person}{M.
  Arcak}.} \bibinfo{year}{2015}\natexlab{}.
\newblock \showarticletitle{Efficient finite abstraction of mixed monotone
  systems}. In \bibinfo{booktitle}{\emph{Proceedings of the 18th International
  Conference on Hybrid Systems: Computation and Control}}. ACM,
  \bibinfo{pages}{58--67}.
\newblock


\bibitem[\protect\citeauthoryear{Dreossi, Donz{\'e}, and Seshia}{Dreossi
  et~al\mbox{.}}{2017}]%
        {dreossi17}
\bibfield{author}{\bibinfo{person}{T. Dreossi}, \bibinfo{person}{A. Donz{\'e}},
  {and} \bibinfo{person}{S.~A. Seshia}.} \bibinfo{year}{2017}\natexlab{}.
\newblock \showarticletitle{Compositional falsification of cyber-physical
  systems with machine learning components}. In \bibinfo{booktitle}{\emph{NASA
  Formal Methods Symposium}}. Springer, \bibinfo{pages}{357--372}.
\newblock


\bibitem[\protect\citeauthoryear{Dutta, Jha, Sankaranarayanan, and
  Tiwari}{Dutta et~al\mbox{.}}{2018}]%
        {dutta18}
\bibfield{author}{\bibinfo{person}{S. Dutta}, \bibinfo{person}{S. Jha},
  \bibinfo{person}{S. Sankaranarayanan}, {and} \bibinfo{person}{A. Tiwari}.}
  \bibinfo{year}{2018}\natexlab{}.
\newblock \showarticletitle{Output Range Analysis for Deep Feedforward Neural
  Networks}. In \bibinfo{booktitle}{\emph{NASA Formal Methods Symposium}}.
  Springer, \bibinfo{pages}{121--138}.
\newblock


\bibitem[\protect\citeauthoryear{Ehlers}{Ehlers}{2017}]%
        {ehlers17}
\bibfield{author}{\bibinfo{person}{R. Ehlers}.}
  \bibinfo{year}{2017}\natexlab{}.
\newblock \showarticletitle{Formal verification of piece-wise linear
  feed-forward neural networks}. In \bibinfo{booktitle}{\emph{International
  Symposium on Automated Technology for Verification and Analysis}}. Springer,
  \bibinfo{pages}{269--286}.
\newblock


\bibitem[\protect\citeauthoryear{Frehse, Guernic, Donz{\'e}, Cotton, Ray,
  Lebeltel, Ripado, Girard, Dang, and Maler}{Frehse et~al\mbox{.}}{2011}]%
        {frehse11}
\bibfield{author}{\bibinfo{person}{G. Frehse}, \bibinfo{person}{C.~L. Guernic},
  \bibinfo{person}{A. Donz{\'e}}, \bibinfo{person}{S. Cotton},
  \bibinfo{person}{R. Ray}, \bibinfo{person}{O. Lebeltel}, \bibinfo{person}{R.
  Ripado}, \bibinfo{person}{A. Girard}, \bibinfo{person}{T. Dang}, {and}
  \bibinfo{person}{O. Maler}.} \bibinfo{year}{2011}\natexlab{}.
\newblock \showarticletitle{{SpaceEx: Scalable verification of hybrid
  systems}}. In \bibinfo{booktitle}{\emph{International Conference on Computer
  Aided Verification}}. \bibinfo{pages}{379--395}.
\newblock


\bibitem[\protect\citeauthoryear{Gao, Kong, Chen, and Clarke}{Gao
  et~al\mbox{.}}{2014}]%
        {gao14}
\bibfield{author}{\bibinfo{person}{S. Gao}, \bibinfo{person}{S. Kong},
  \bibinfo{person}{W. Chen}, {and} \bibinfo{person}{E. Clarke}.}
  \bibinfo{year}{2014}\natexlab{}.
\newblock \showarticletitle{Delta-complete analysis for bounded reachability of
  hybrid systems}.
\newblock \bibinfo{journal}{\emph{arXiv preprint arXiv:1404.7171}}
  (\bibinfo{year}{2014}).
\newblock


\bibitem[\protect\citeauthoryear{Goodfellow, Bengio, Courville, and
  Bengio}{Goodfellow et~al\mbox{.}}{2016}]%
        {goodfellow16}
\bibfield{author}{\bibinfo{person}{I. Goodfellow}, \bibinfo{person}{Y. Bengio},
  \bibinfo{person}{A. Courville}, {and} \bibinfo{person}{Y. Bengio}.}
  \bibinfo{year}{2016}\natexlab{}.
\newblock \bibinfo{booktitle}{\emph{Deep learning}}. Vol.~\bibinfo{volume}{1}.
\newblock \bibinfo{publisher}{MIT press Cambridge}.
\newblock


\bibitem[\protect\citeauthoryear{Hertneck, K{\"o}hler, Trimpe, and
  Allg{\"o}wer}{Hertneck et~al\mbox{.}}{2018}]%
        {hertneck18}
\bibfield{author}{\bibinfo{person}{M. Hertneck}, \bibinfo{person}{J.
  K{\"o}hler}, \bibinfo{person}{S. Trimpe}, {and} \bibinfo{person}{F.
  Allg{\"o}wer}.} \bibinfo{year}{2018}\natexlab{}.
\newblock \showarticletitle{Learning an approximate model predictive controller
  with guarantees}.
\newblock \bibinfo{journal}{\emph{IEEE Control Systems Letters}}
  \bibinfo{volume}{2}, \bibinfo{number}{3} (\bibinfo{year}{2018}),
  \bibinfo{pages}{543--548}.
\newblock


\bibitem[\protect\citeauthoryear{Hornik, Stinchcombe, and White}{Hornik
  et~al\mbox{.}}{1989}]%
        {hornik89}
\bibfield{author}{\bibinfo{person}{K. Hornik}, \bibinfo{person}{M.
  Stinchcombe}, {and} \bibinfo{person}{H. White}.}
  \bibinfo{year}{1989}\natexlab{}.
\newblock \showarticletitle{Multilayer feedforward networks are universal
  approximators}.
\newblock \bibinfo{journal}{\emph{Neural networks}} \bibinfo{volume}{2},
  \bibinfo{number}{5} (\bibinfo{year}{1989}), \bibinfo{pages}{359--366}.
\newblock


\bibitem[\protect\citeauthoryear{Julian, Lopez, Brush, Owen, and
  Kochenderfer}{Julian et~al\mbox{.}}{2016}]%
        {julian16}
\bibfield{author}{\bibinfo{person}{K.~D. Julian}, \bibinfo{person}{J. Lopez},
  \bibinfo{person}{J.~S. Brush}, \bibinfo{person}{M.~P. Owen}, {and}
  \bibinfo{person}{M.~J. Kochenderfer}.} \bibinfo{year}{2016}\natexlab{}.
\newblock \showarticletitle{Policy compression for aircraft collision avoidance
  systems}. In \bibinfo{booktitle}{\emph{Digital Avionics Systems Conference
  (DASC), 2016 IEEE/AIAA 35th}}. IEEE, \bibinfo{pages}{1--10}.
\newblock


\bibitem[\protect\citeauthoryear{Katz, Barrett, Dill, Julian, and
  Kochenderfer}{Katz et~al\mbox{.}}{2017}]%
        {katz17}
\bibfield{author}{\bibinfo{person}{G. Katz}, \bibinfo{person}{C. Barrett},
  \bibinfo{person}{D.~L. Dill}, \bibinfo{person}{K. Julian}, {and}
  \bibinfo{person}{M.~J. Kochenderfer}.} \bibinfo{year}{2017}\natexlab{}.
\newblock \showarticletitle{Reluplex: An efficient SMT solver for verifying
  deep neural networks}. In \bibinfo{booktitle}{\emph{International Conference
  on Computer Aided Verification}}. Springer, \bibinfo{pages}{97--117}.
\newblock


\bibitem[\protect\citeauthoryear{Kearns and Vazirani}{Kearns and
  Vazirani}{1994}]%
        {kearns94}
\bibfield{author}{\bibinfo{person}{M.~J. Kearns} {and} \bibinfo{person}{U.
  Vazirani}.} \bibinfo{year}{1994}\natexlab{}.
\newblock \bibinfo{booktitle}{\emph{An introduction to computational learning
  theory}}.
\newblock \bibinfo{publisher}{MIT press}.
\newblock


\bibitem[\protect\citeauthoryear{Kong, Gao, Chen, and Clarke}{Kong
  et~al\mbox{.}}{2015}]%
        {kong15}
\bibfield{author}{\bibinfo{person}{S. Kong}, \bibinfo{person}{S. Gao},
  \bibinfo{person}{W. Chen}, {and} \bibinfo{person}{E. Clarke}.}
  \bibinfo{year}{2015}\natexlab{}.
\newblock \showarticletitle{dReach: $\delta$-reachability analysis for hybrid
  systems}. In \bibinfo{booktitle}{\emph{International Conference on TOOLS and
  Algorithms for the Construction and Analysis of Systems}}. Springer,
  \bibinfo{pages}{200--205}.
\newblock


\bibitem[\protect\citeauthoryear{Lafferriere, Pappas, and Yovine}{Lafferriere
  et~al\mbox{.}}{1999}]%
        {lafferriere99}
\bibfield{author}{\bibinfo{person}{G. Lafferriere}, \bibinfo{person}{G.~J.
  Pappas}, {and} \bibinfo{person}{S. Yovine}.} \bibinfo{year}{1999}\natexlab{}.
\newblock \showarticletitle{A new class of decidable hybrid systems}. In
  \bibinfo{booktitle}{\emph{International Workshop on Hybrid Systems:
  Computation and Control}}. \bibinfo{pages}{137--151}.
\newblock


\bibitem[\protect\citeauthoryear{Lillicrap, Hunt, Pritzel, Heess, Erez, Tassa,
  Silver, and Wierstra}{Lillicrap et~al\mbox{.}}{2015}]%
        {lillicrap15}
\bibfield{author}{\bibinfo{person}{T.~P. Lillicrap}, \bibinfo{person}{J.~J.
  Hunt}, \bibinfo{person}{A. Pritzel}, \bibinfo{person}{N. Heess},
  \bibinfo{person}{T. Erez}, \bibinfo{person}{Y. Tassa}, \bibinfo{person}{D.
  Silver}, {and} \bibinfo{person}{D. Wierstra}.}
  \bibinfo{year}{2015}\natexlab{}.
\newblock \showarticletitle{Continuous control with deep reinforcement
  learning}.
\newblock \bibinfo{journal}{\emph{arXiv preprint arXiv:1509.02971}}
  (\bibinfo{year}{2015}).
\newblock


\bibitem[\protect\citeauthoryear{Mnih, Kavukcuoglu, Silver, Rusu, Veness,
  Bellemare, Graves, Riedmiller, Fidjeland, Ostrovski, et~al\mbox{.}}{Mnih
  et~al\mbox{.}}{2015}]%
        {mnih15}
\bibfield{author}{\bibinfo{person}{V. Mnih}, \bibinfo{person}{K. Kavukcuoglu},
  \bibinfo{person}{D. Silver}, \bibinfo{person}{A.~A. Rusu},
  \bibinfo{person}{J. Veness}, \bibinfo{person}{M.~G. Bellemare},
  \bibinfo{person}{A. Graves}, \bibinfo{person}{M. Riedmiller},
  \bibinfo{person}{A.~K. Fidjeland}, \bibinfo{person}{G. Ostrovski},
  {et~al\mbox{.}}} \bibinfo{year}{2015}\natexlab{}.
\newblock \showarticletitle{Human-level control through deep reinforcement
  learning}.
\newblock \bibinfo{journal}{\emph{Nature}} \bibinfo{volume}{518},
  \bibinfo{number}{7540} (\bibinfo{year}{2015}), \bibinfo{pages}{529}.
\newblock


\bibitem[\protect\citeauthoryear{Mohri, Rostamizadeh, and Talwalkar}{Mohri
  et~al\mbox{.}}{2012}]%
        {mohri12}
\bibfield{author}{\bibinfo{person}{M. Mohri}, \bibinfo{person}{A.
  Rostamizadeh}, {and} \bibinfo{person}{A. Talwalkar}.}
  \bibinfo{year}{2012}\natexlab{}.
\newblock \bibinfo{booktitle}{\emph{Foundations of machine learning}}.
\newblock \bibinfo{publisher}{MIT press}.
\newblock


\bibitem[\protect\citeauthoryear{Moura and Bj{\o}rner}{Moura and
  Bj{\o}rner}{2008}]%
        {moura08}
\bibfield{author}{\bibinfo{person}{L.~D. Moura} {and} \bibinfo{person}{N.
  Bj{\o}rner}.} \bibinfo{year}{2008}\natexlab{}.
\newblock \showarticletitle{Z3: An efficient SMT solver}. In
  \bibinfo{booktitle}{\emph{International conference on Tools and Algorithms
  for the Construction and Analysis of Systems}}. Springer,
  \bibinfo{pages}{337--340}.
\newblock


\bibitem[\protect\citeauthoryear{OpenAI}{OpenAI}{[n. d.]}]%
        {openaigym}
\bibfield{author}{\bibinfo{person}{OpenAI}.} \bibinfo{year}{[n.
  d.]}\natexlab{}.
\newblock \bibinfo{title}{OpenAI Gym}.
\newblock
\newblock
\newblock
\shownote{https://gym.openai.com.}


\bibitem[\protect\citeauthoryear{Optimization}{Optimization}{[n. d.]}]%
        {gurobi}
\bibfield{author}{\bibinfo{person}{Gurobi Optimization}.} \bibinfo{year}{[n.
  d.]}\natexlab{}.
\newblock \bibinfo{title}{Gurobi Optimizer}.
\newblock
\newblock
\newblock
\shownote{https://gurobi.com.}


\bibitem[\protect\citeauthoryear{Rolnick and Tegmark}{Rolnick and
  Tegmark}{2017}]%
        {rolnick17}
\bibfield{author}{\bibinfo{person}{D. Rolnick} {and} \bibinfo{person}{M.
  Tegmark}.} \bibinfo{year}{2017}\natexlab{}.
\newblock \showarticletitle{The power of deeper networks for expressing natural
  functions}.
\newblock \bibinfo{journal}{\emph{arXiv preprint arXiv:1705.05502}}
  (\bibinfo{year}{2017}).
\newblock


\bibitem[\protect\citeauthoryear{Royo, Fridovich-Keil, Herbert, and
  Tomlin}{Royo et~al\mbox{.}}{2018}]%
        {royo18}
\bibfield{author}{\bibinfo{person}{V.~R. Royo}, \bibinfo{person}{D.
  Fridovich-Keil}, \bibinfo{person}{S. Herbert}, {and} \bibinfo{person}{C.~J.
  Tomlin}.} \bibinfo{year}{2018}\natexlab{}.
\newblock \showarticletitle{Classification-based Approximate Reachability with
  Guarantees Applied to Safe Trajectory Tracking}.
\newblock \bibinfo{journal}{\emph{arXiv preprint arXiv:1803.03237}}
  (\bibinfo{year}{2018}).
\newblock


\bibitem[\protect\citeauthoryear{Silver, Huang, Maddison, Guez,
  et~al\mbox{.}}{Silver et~al\mbox{.}}{2016}]%
        {silver16}
\bibfield{author}{\bibinfo{person}{D. Silver}, \bibinfo{person}{A. Huang},
  \bibinfo{person}{C.~J. Maddison}, \bibinfo{person}{A. Guez}, {et~al\mbox{.}}}
  \bibinfo{year}{2016}\natexlab{}.
\newblock \showarticletitle{Mastering the game of Go with deep neural networks
  and tree search}.
\newblock \bibinfo{journal}{\emph{nature}} \bibinfo{volume}{529},
  \bibinfo{number}{7587} (\bibinfo{year}{2016}), \bibinfo{pages}{484}.
\newblock


\bibitem[\protect\citeauthoryear{Sutskever, Vinyals, and Le}{Sutskever
  et~al\mbox{.}}{2014}]%
        {sutskever14}
\bibfield{author}{\bibinfo{person}{I. Sutskever}, \bibinfo{person}{O. Vinyals},
  {and} \bibinfo{person}{Q. Le}.} \bibinfo{year}{2014}\natexlab{}.
\newblock \showarticletitle{Sequence to sequence learning with neural
  networks}. In \bibinfo{booktitle}{\emph{Advances in neural information
  processing systems}}. \bibinfo{pages}{3104--3112}.
\newblock


\bibitem[\protect\citeauthoryear{Szegedy, Zaremba, Sutskever, Bruna, Erhan,
  et~al\mbox{.}}{Szegedy et~al\mbox{.}}{2013}]%
        {szegedy13}
\bibfield{author}{\bibinfo{person}{C. Szegedy}, \bibinfo{person}{W. Zaremba},
  \bibinfo{person}{I. Sutskever}, \bibinfo{person}{J. Bruna},
  \bibinfo{person}{D. Erhan}, {et~al\mbox{.}}} \bibinfo{year}{2013}\natexlab{}.
\newblock \showarticletitle{Intriguing properties of neural networks}.
\newblock \bibinfo{journal}{\emph{arXiv preprint arXiv:1312.6199}}
  (\bibinfo{year}{2013}).
\newblock


\bibitem[\protect\citeauthoryear{Taigman, Yang, Ranzato, and Wolf}{Taigman
  et~al\mbox{.}}{2014}]%
        {taigman14}
\bibfield{author}{\bibinfo{person}{Y. Taigman}, \bibinfo{person}{M. Yang},
  \bibinfo{person}{M. Ranzato}, {and} \bibinfo{person}{L. Wolf}.}
  \bibinfo{year}{2014}\natexlab{}.
\newblock \showarticletitle{Deepface: Closing the gap to human-level
  performance in face verification}. In \bibinfo{booktitle}{\emph{Proceedings
  of the IEEE conference on computer vision and pattern recognition}}.
  \bibinfo{pages}{1701--1708}.
\newblock


\bibitem[\protect\citeauthoryear{Tarski}{Tarski}{1998}]%
        {tarski98}
\bibfield{author}{\bibinfo{person}{A. Tarski}.}
  \bibinfo{year}{1998}\natexlab{}.
\newblock \showarticletitle{A decision method for elementary algebra and
  geometry}.
\newblock In \bibinfo{booktitle}{\emph{Quantifier elimination and cylindrical
  algebraic decomposition}}. \bibinfo{publisher}{Springer},
  \bibinfo{pages}{24--84}.
\newblock


\bibitem[\protect\citeauthoryear{Telgarsky}{Telgarsky}{2016}]%
        {telgarsky16}
\bibfield{author}{\bibinfo{person}{M. Telgarsky}.}
  \bibinfo{year}{2016}\natexlab{}.
\newblock \showarticletitle{Benefits of depth in neural networks}.
\newblock \bibinfo{journal}{\emph{arXiv preprint arXiv:1602.04485}}
  (\bibinfo{year}{2016}).
\newblock


\bibitem[\protect\citeauthoryear{Vanderbei et~al\mbox{.}}{Vanderbei
  et~al\mbox{.}}{2015}]%
        {vanderbei15}
\bibfield{author}{\bibinfo{person}{R.~J. Vanderbei} {et~al\mbox{.}}}
  \bibinfo{year}{2015}\natexlab{}.
\newblock \bibinfo{booktitle}{\emph{Linear programming}}.
\newblock \bibinfo{publisher}{Springer}.
\newblock


\bibitem[\protect\citeauthoryear{Wilkie}{Wilkie}{1997}]%
        {wilkie97}
\bibfield{author}{\bibinfo{person}{A.~J. Wilkie}.}
  \bibinfo{year}{1997}\natexlab{}.
\newblock \showarticletitle{Schanuel’s conjecture and the decidability of the
  real exponential field}.
\newblock In \bibinfo{booktitle}{\emph{Algebraic Model Theory}}.
  \bibinfo{publisher}{Springer}, \bibinfo{pages}{223--230}.
\newblock


\bibitem[\protect\citeauthoryear{Xiang, Musau, Wild, Lopez, Hamilton, Yang,
  Rosenfeld, and Johnson}{Xiang et~al\mbox{.}}{2018}]%
        {xiang18}
\bibfield{author}{\bibinfo{person}{W. Xiang}, \bibinfo{person}{P. Musau},
  \bibinfo{person}{A.~A. Wild}, \bibinfo{person}{D.~M. Lopez},
  \bibinfo{person}{N. Hamilton}, \bibinfo{person}{X. Yang}, \bibinfo{person}{J.
  Rosenfeld}, {and} \bibinfo{person}{T.~T Johnson}.}
  \bibinfo{year}{2018}\natexlab{}.
\newblock \showarticletitle{Verification for Machine Learning, Autonomy, and
  Neural Networks Survey}.
\newblock \bibinfo{journal}{\emph{arXiv preprint arXiv:1810.01989}}
  (\bibinfo{year}{2018}).
\newblock


\bibitem[\protect\citeauthoryear{Xiang, Tran, and Johnson}{Xiang
  et~al\mbox{.}}{2017}]%
        {xiang17}
\bibfield{author}{\bibinfo{person}{W. Xiang}, \bibinfo{person}{H.~D. Tran},
  {and} \bibinfo{person}{T.~T. Johnson}.} \bibinfo{year}{2017}\natexlab{}.
\newblock \showarticletitle{Output reachable set estimation and verification
  for multi-layer neural networks}.
\newblock \bibinfo{journal}{\emph{arXiv preprint arXiv:1708.03322}}
  (\bibinfo{year}{2017}).
\newblock


\bibitem[\protect\citeauthoryear{Zhang, Bengio, Hardt, Recht, and
  Vinyals}{Zhang et~al\mbox{.}}{2016}]%
        {zhang16}
\bibfield{author}{\bibinfo{person}{C. Zhang}, \bibinfo{person}{S. Bengio},
  \bibinfo{person}{M. Hardt}, \bibinfo{person}{B. Recht}, {and}
  \bibinfo{person}{O. Vinyals}.} \bibinfo{year}{2016}\natexlab{}.
\newblock \showarticletitle{Understanding deep learning requires rethinking
  generalization}.
\newblock \bibinfo{journal}{\emph{arXiv preprint arXiv:1611.03530}}
  (\bibinfo{year}{2016}).
\newblock


\end{thebibliography}

\end{document}